\documentclass[journal]{IEEEtran}
\usepackage{amsmath,amssymb,amsthm,bm}
\usepackage{graphicx}
\usepackage[caption=false,font=footnotesize]{subfig}
\usepackage{cite}
\usepackage{booktabs}
\usepackage[hidelinks]{hyperref}
\newtheorem{theorem}{Theorem}
\newtheorem{lemma}{Lemma}
\newtheorem{corollary}{Corollary}
\newtheorem{remark}{Remark}

\title{Performance Analysis of RSMA-Enabled Bistatic ISAC in LEO Networks with Holographic Apertures and Fluid-Antenna Users}
\author{Wali Ullah Khan, Chandan Kumar Sheemar, Muhammad Adil, Symeon Chatzinotas\thanks{Wali Ullah Khan, Chandan Kumar Sheemar and Symeon Chatzinotas are with the Interdisciplinary Centre for Security, Reliability, and Trust (SnT), University of Luxembourg, Luxembourg (e-mails: waliullahkhan30@gmail.com, chandankumar.sheemar@uni.lu, symeon.chatzinotas@uni.lu).

Muhammad Adil is with the Department of Electronics Engineering, University of Rome Tor Vergata, 00133 Rome, Italy (e-mail: muhammad.adil@uniroma2.it).

}}

\begin{document}
\maketitle

\begin{abstract}
This paper develops an ergodic performance framework for rate-splitting multiple access (RSMA)-enabled bistatic integrated sensing and communication (ISAC) in a low-Earth-orbit (LEO) satellite network with an amplitude-constrained reconfigurable holographic surface (RHS) and fluid-antenna-system (FAS) users. Deterministic angle-based common and zero-forcing private reference beams are realized through one shared multi-feed RHS amplitude state and stream-specific feed-domain precoders, and the resulting self-, leakage-, and target-direction gains are retained explicitly. Conservative private- and common-rate lower bounds are derived for both reference-port and best-of-$P$ FAS reception while preserving the same-port selection coupling. For sensing, a closed-form average bistatic sensing signal-to-noise ratio (SNR) is obtained under nearest-receiver association and a finite target--receiver guard distance, with extensions to angle-conditioned footprint averaging and angular scheduling. Monte Carlo results confirm the tightness of the analytical rate bounds and validate the sensing expressions. Benchmarks show that scalar RHS-efficiency models can miss strong direction-dependent effects and that nearest-ground-receiver bistatic sensing provides a $17.7$--$25.7$~dB mean SNR advantage over a favorable monostatic LEO reference for $N_{\rm RHS}=16384$ over LEO altitudes of $400$--$1000$~km. FAS gains are largest in scattering-rich regimes, while the realized shared-state RHS target gain need not vary monotonically with aperture size.
\end{abstract}


 \section{Introduction}

Integrated sensing and communication (ISAC) is emerging as a key capability for non-terrestrial networks (NTNs), where satellite waveforms can simultaneously provide wide-area connectivity and illuminate geographically distributed targets \cite{10989572,11143190,11488929}. Low-Earth-orbit (LEO) satellites are particularly attractive for ISAC because their lower orbital altitude reduces propagation delay and satellite-to-ground path loss compared with higher-orbit platforms \cite{10981514}. However, the long propagation distances still require large transmit apertures to provide sufficient link and sensing gains. Conventional fully active phased arrays can therefore incur considerable radio-frequency (RF), phase-shifter, hardware-complexity, and power-consumption overhead as the aperture size increases \cite{11157894,11328800,11357516}. Reconfigurable holographic surfaces (RHSs) offer a promising alternative by realizing electrically large apertures through a guided reference wave that excites a dense set of low-power meta-elements \cite{11452230}. Unlike ideal complex-weight arrays, however, practical RHSs operate over a constrained excitation manifold. Consequently, realizing multiple desired beams through one shared amplitude-constrained holographic transfer matrix modifies not only their intended directional gains but also the residual leakage toward other users and sensing directions \cite{11037747}.

The resulting directional leakage becomes particularly relevant in multiuser satellite links, where users share limited spectral and spatial resources. Rate-splitting multiple access (RSMA) has emerged as an effective interference-management strategy by dividing user messages into common and private parts, allowing the common stream to partially decode interference while the private streams preserve user-specific information \cite{10097680,10684731,10312769,10896843}. RSMA is also naturally suited to satellite ISAC because the common stream can simultaneously act as a shared sensing-illumination waveform, thereby coupling communication and sensing through its power allocation and spatial beam direction \cite{10989572,11143190}. Although recent LEO-ISAC studies have investigated hybrid precoding and bistatic RSMA architectures \cite{10989572,11143190}, they typically rely on conventional array models and therefore do not characterize the direction-dependent self-gain and inter-user leakage introduced by amplitude-constrained holographic beam synthesis.

At the user side, fluid antenna systems (FASs) provide an additional spatial degree of freedom by selecting a favorable receive position among multiple correlated candidate ports distributed over a compact aperture \cite{10753482,11302793,10599127,11175437}. Their potential has recently been explored in satellite and NTN communications \cite{11622450,11482755,11194137,11184548,11644016}. In RSMA reception, however, the common and private streams must be decoded at the same selected physical port. Hence, selecting the port according to the desired private-stream gain simultaneously alters the interference observed at that port and the common-stream channel available before successive interference cancellation (SIC). This statistical coupling becomes especially relevant with an RHS transmitter because projection through a shared amplitude-constrained holographic transfer matrix generally destroys the orthogonality of the ideal zero-forcing (ZF) private beams.

The sensing architecture introduces a further challenge. In monostatic LEO sensing, the target echo experiences two satellite-scale propagation legs, resulting in a severe two-way path-loss penalty. Bistatic sensing alleviates this limitation by using the LEO satellite as the illuminator while a spatially separated terrestrial receiver collects the target echo \cite{10788035,11143190,11488929}. With sensing-only ground base stations (BSs) distributed over the service region, a target can be associated with its nearest receiver, replacing the long satellite return path with a much shorter terrestrial sensing hop. These intertwined effects motivate a unified performance analysis that jointly captures amplitude-constrained RHS radiation, RSMA interference management, FAS port selection and its induced statistical coupling, bistatic LEO sensing geometry, and the spatial distribution of terrestrial sensing receivers.

\subsection{Related Work and Research Gap}
The closest literature falls into four complementary directions. First, RSMA-enabled satellite and LEO-ISAC systems have demonstrated the value of common/private message splitting and shared sensing illumination \cite{11038752,10989572,11143190}. Second, RHS and holographic-metasurface architectures have been studied for satellite links to reduce large-aperture RF complexity and to enable communication-oriented hybrid or holographic beamforming \cite{10925876,10844052,11141770,11271825}. Third, bistatic LEO-ISAC architectures separate the satellite illuminator from a terrestrial echo receiver and optimize communication--sensing beamforming under geometric or channel constraints \cite{11143190,11488929,11284853}. Fourth, FAS-assisted satellite communications exploit correlated spatial port selection to improve outage, ergodic rate, or receiver compactness \cite{11482755,11194137,11184548,11644016}.

These lines of work do not yet provide a unified analytical characterization of an RSMA-enabled bistatic LEO-ISAC link in which (i) all RSMA streams are realized through one shared amplitude-constrained multi-feed holographic state, (ii) the realized post-projection self and leakage gains are retained rather than replaced by an ideal array-gain factor, (iii) both RSMA decoding stages are evaluated at the same FAS-selected port while accounting for selection-induced interference coupling, and (iv) the bistatic echo is received by the nearest node of a spatially distributed sensing layer. The analytical difficulty is not only the coexistence of these components: the shared-RHS projection generally destroys the ideal ZF orthogonality, FAS selection couples numerator and denominator statistics at the selected port, and the scheduled user angles are statistically dependent on the path-loss geometry. These effects must therefore be handled explicitly rather than by independent substitutions of conventional results.

Motivated by this gap, we develop an ergodic performance-analysis framework for an RHS-equipped LEO satellite serving FAS users through RSMA while using the common stream for bistatic target illumination. The main contributions are:
\begin{itemize}
\item An amplitude-constrained multi-feed RHS model with deterministic angle- and ephemeris-based RSMA precoding is proposed. A single real-valued holographic amplitude state is shared by the common and all private streams, while stream separation is performed in the feed domain. The resulting user self-gains, cross-user leakage, target gain, and realized radiated powers are retained explicitly.
\item We derive conservative ergodic private- and common-rate bounds from Rician log moments without assuming desired/interference independence. For FAS reception, best-of-$P$ self-gain statistics, selected-port interference moments, and selected-port common-stream log moments preserve the requirement that both SIC stages use the same selected port.
\item We derive a closed-form average bistatic sensing SNR under nearest-ground-receiver association and a finite guard distance, together with angle-conditioned footprint averaging, the induced angular candidate process, and a scheduling-based ZF-gain guarantee. Monte Carlo validation and benchmarks against ideal/scalar RHS models, random/exhaustive scheduling, and a favorable monostatic reference quantify the resulting tradeoffs.
\end{itemize}

The remainder of the paper is organized as follows. Section~\ref{sec:system_model} presents the network, RHS, channel, RSMA, scheduling, and bistatic sensing models. Section~\ref{sec:analysis} develops the ergodic communication and sensing analysis and its spatial/angular extensions. Section~\ref{sec:numerical_results} validates the analysis and quantifies the principal communication--sensing trends. Section~\ref{sec:conclusion} concludes the paper.

\textit{Notation:}
Boldface lower- and upper-case letters denote vectors and matrices, respectively. $(\cdot)^T$, $(\cdot)^H$, and $(\cdot)^*$ denote transpose, Hermitian transpose, and complex conjugation; $\|\cdot\|_2$ and $|\cdot|$ denote the Euclidean norm and scalar magnitude. $\Re\{\cdot\}$, $\mathbb E[\cdot]$, and $\Pr[\cdot]$ denote real part, expectation, and probability. $\mathcal{CN}(\boldsymbol\mu,\mathbf C)$ and $\mathcal N(\mu,\sigma^2)$ denote circularly symmetric complex Gaussian and real Gaussian distributions, respectively. $\mathbf I_N$ and $\mathbf 1_N$ denote the $N\times N$ identity matrix and the $N$-dimensional all-ones vector, $\otimes$ denotes the Kronecker product, $j\triangleq\sqrt{-1}$, $\lambda$ is the carrier wavelength, and $k_0\triangleq2\pi/\lambda$ is the free-space wavenumber.

\begin{figure}[t]
\centering   
\includegraphics[width=\columnwidth]{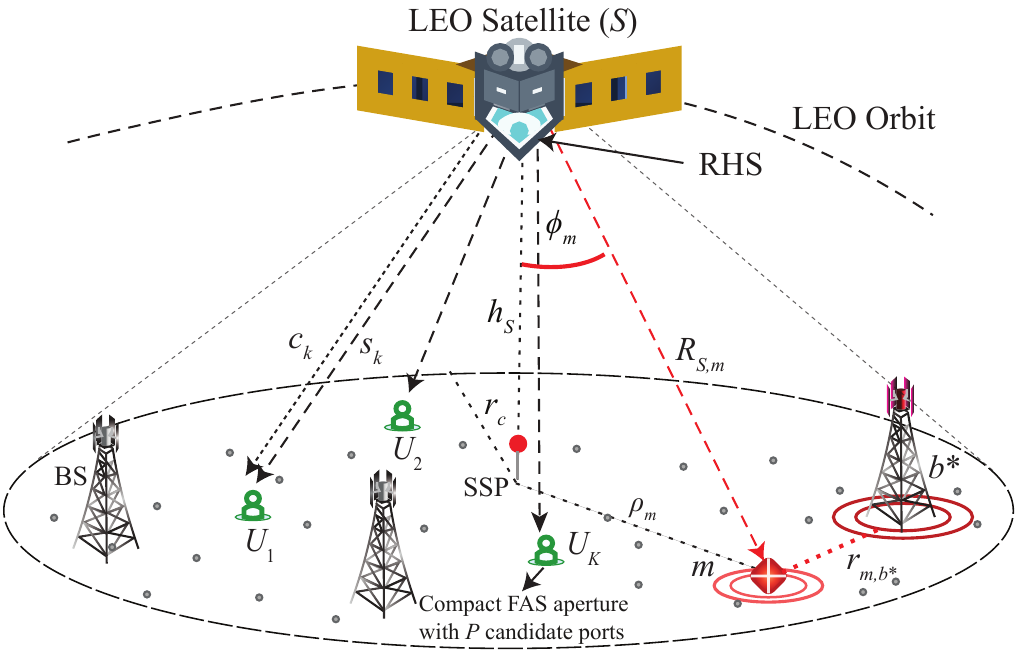}
\caption{System geometry for the LEO RSMA-enabled bistatic ISAC model. Candidate users lie in the satellite coverage footprint; scheduled FAS users receive the RSMA streams, while a representative target is illuminated by the common stream and associated with its nearest sensing-only BS.}
\label{fig:system_geometry}
\end{figure}

\section{System Model}
\label{sec:system_model}

\subsection{Network Topology}
As illustrated in Fig.~\ref{fig:system_geometry}, we consider a quasi-static downlink snapshot of a LEO satellite $S$ at altitude $h_s$. Let \(\mathcal D=\{(x,y):x^2+y^2\le r_c^2\}\) denote the circular satellite coverage footprint. Candidate users and sensing targets are the restrictions to $\mathcal D$ of two independent homogeneous Poisson point processes (PPPs) $\Phi_u$ and $\Phi_r$ with densities $\lambda_u$ and $\lambda_r$. The sensing-only BS process $\Phi_b$, of density $\lambda_b$, is modeled on the terrestrial plane (equivalently, on an infrastructure region much larger than $\mathcal D$) and is independent of $\Phi_u$ and $\Phi_r$. This distinction is deliberate: the satellite footprint limits the candidate users/targets served or illuminated by the satellite, whereas the nearest terrestrial sensing receiver need not lie inside that footprint. For the single-target analysis we condition on a representative target $m\in\Phi_r\cap\mathcal D$; $\lambda_r$ therefore does not enter the conditional single-target metric explicitly. A ground point at horizontal distance $\rho$ from
the sub-satellite point (SSP) has slant range
\begin{equation}
R(\rho) = \sqrt{h_s^2 + \rho^2},
\end{equation}
under a flat-Earth planar approximation valid at the footprint scale considered (the same ``altitude offset over a planar PPP'' device used for the sensing-target altitude in prior terrestrial ISAC analyses, here applied to the satellite-to-ground geometry instead).

The system comprises four node types with distinct roles:
\begin{itemize}
\item \textbf{Satellite $S$} (the only transmitter): carries an $N_x\times N_y$ reconfigurable holographic surface with $N_{\rm RHS}=N_xN_y$ meta-elements and runs RSMA to $K$ scheduled users. The principal-plane analysis below uses $N_h\equiv N_x$, while the coherent orthogonal-dimension factor $N_y$ is retained explicitly in directional power gains. The common stream also serves as the sensing illumination waveform.
\item \textbf{Ground users} $k \in \{1,\ldots,K\} \subset \Phi_u$: FAS-equipped RSMA receivers.
\item \textbf{Ground BSs} $b \in \Phi_b$: sensing-only passive receivers, each with $N_r$ fixed antennas; not communication nodes in this model.
\item \textbf{Targets} $m \in \Phi_r$: static ground/infrastructure reflectors with dimensionless reflection factor $\zeta_m$; Remark~\ref{rem:rcs_normalization} relates it to the physical radar cross section (RCS) $\sigma_m$ in $\mathrm{m}^2$.
\end{itemize}

We analyze a quasi-static snapshot and assume that satellite-motion-induced Doppler is compensated over the considered interval by standard ephemeris-assisted synchronization; residual Doppler is therefore omitted from the performance expressions.

\subsection{Reconfigurable Holographic Aperture Model}
\label{sec:rhs}

The satellite employs a multi-feed RHS in which a single physical amplitude state is shared by all simultaneously transmitted streams. Let $N_f\ge K+1$ denote the number of active feed ports and let
\begin{equation}
\mathbf\Psi=[\boldsymbol\psi_1,\ldots,\boldsymbol\psi_{N_f}]\in\mathbb C^{N_h\times N_f}
\label{eq:feed_matrix}
\end{equation}
collect the deterministic guided reference-wave responses from the feeds to the $N_h$ principal-plane meta-elements. Its $(n,f)$th entry may be written generically as $[\mathbf\Psi]_{n,f}=\eta_{n,f}e^{-jk_g d_{n,f}}$, where $d_{n,f}$ is the guided-wave distance, $\eta_{n,f}$ captures deterministic feed-to-element attenuation, and $k_g$ denotes the guided-wave wavenumber~\cite{11452230,10844052}. The matrix $\mathbf\Psi$ is a fixed hardware coupling matrix and is not redesigned when the user/target geometry or other swept parameters change. The numerical study adopts $k_g/k_0=1$.

Let
\begin{equation}
\boldsymbol\nu=[\nu_1,\ldots,\nu_{N_h}]^T,\qquad 0\le\nu_n\le A_{\max},
\label{eq:shared_state}
\end{equation}
be the \emph{single shared} holographic amplitude state. The corresponding analog transfer matrix is
\begin{equation}
\mathbf F(\boldsymbol\nu)=\operatorname{diag}(\boldsymbol\nu)\mathbf\Psi\in\mathbb C^{N_h\times N_f}.
\label{eq:shared_rhs_transfer}
\end{equation}
Thus, independent stream-wise amplitude patterns are not available: all RSMA streams experience the same $\mathbf F(\boldsymbol\nu)$ and differ only through their feed-domain precoders.

Let $\theta$ denote the principal-plane departure angle measured from boresight. The far-field steering response toward $\theta$ is
\begin{equation}
\begin{aligned}
\mathbf a_t(\theta)&=\big[e^{j\frac{2\pi}{\lambda}\triangle_t^1(\theta)},\ldots,\\
&\hspace{2.8em}e^{j\frac{2\pi}{\lambda}\triangle_t^{N_h}(\theta)}\big]^T\in\mathbb C^{N_h\times1},
\end{aligned}
\end{equation}
where $\triangle_t^n(\theta)$ is the propagation-distance difference of element $n$ relative to a reference element. Since every steering entry has unit magnitude, $\|\mathbf a_t(\theta)\|_2^2=N_h$. For any realized radiating column $\mathbf q$, the directional field is $\mathbf a_t^H(\theta)\mathbf q$.

\paragraph{Separable planar aperture and power normalization}
The physical transmitter is an $N_x\times N_y$ metasurface with $N_h\equiv N_x$ in the principal-plane analysis. We adopt the separable specialization
\begin{equation}
\mathbf a_{2\rm D}(\theta)=\mathbf a_t(\theta)\otimes\mathbf 1_{N_y},
\qquad
\boldsymbol\psi_y=\frac{\mathbf 1_{N_y}}{\sqrt{N_y}},
\label{eq:2d_sep}
\end{equation}
where $\boldsymbol\psi_y$ is the normalized fixed reference-field factor along the orthogonal dimension. Note that $\mathbf a_{2\rm D}(\theta)$ uses the unnormalized physical steering factor $\mathbf 1_{N_y}$, whereas $\boldsymbol\psi_y$ is normalized to preserve radiated-power normalization. The full two-dimensional shared transfer matrix and realized stream column are therefore
\begin{equation}
\mathbf F_{2\rm D}(\boldsymbol\nu)
=\mathbf F(\boldsymbol\nu)\otimes\boldsymbol\psi_y,
\qquad
\mathbf q_i^{2\rm D}=\mathbf q_i\otimes\boldsymbol\psi_y.
\label{eq:2d_transfer}
\end{equation}
The modulation coefficient $\nu_n$ is replicated across the $N_y$ elements of row $n$ and remains bounded by the same dimensionless hardware limit $A_{\max}$; the factor $1/\sqrt{N_y}$ in $\boldsymbol\psi_y$ normalizes the fixed feed field rather than the modulation state. Consequently,
\begin{equation}
\|\mathbf q_i^{2\rm D}\|_2^2=\|\mathbf q_i\|_2^2=a_i^2,
\label{eq:2d_power_norm}
\end{equation}
so increasing $N_y$ does not inject additional radiated power. For the corresponding unit-norm direction $\bar{\mathbf w}_i^{2\rm D}=\bar{\mathbf w}_i\otimes\boldsymbol\psi_y$,
\begin{equation}
\left|\mathbf a_{2\rm D}^H(\theta)\bar{\mathbf w}_i^{2\rm D}\right|^2
=N_y\left|\mathbf a_t^H(\theta)\bar{\mathbf w}_i\right|^2.
\label{eq:upa_gain}
\end{equation}
Thus the reduced $N_h$-dimensional model preserves the full-aperture radiated-power normalization exactly, while the coherent orthogonal dimension appears only in the deterministic LoS directional gain. The physical aperture contains $N_{\rm RHS}=N_hN_y$ meta-elements, and all reported two-dimensional user/target gains use \eqref{eq:upa_gain}. This is a separable principal-plane specialization, not a claim of a fully general two-dimensional electromagnetic factorization.

\subsection{Channel Model}

\subsubsection{Satellite-to-User Channel}
\label{sec:chan_user}

Ground users experience Rician fading with factor $\kappa_c$. User $k$ has $P$ candidate FAS ports over a normalized aperture $W_{\rm FAS}$ wavelengths. At port $p$ and transmit element $n$,
\begin{equation}
h_{k,p,n}=\sqrt{\frac{\kappa_c}{\kappa_c+1}}e^{j\frac{2\pi}{\lambda}\Delta^n_t(\theta_k)}+\sqrt{\frac{1}{\kappa_c+1}}\widetilde h_{k,p,n}.
\label{eq:fas_channel}
\end{equation}
The compact FAS displacement axis is taken transverse to the dominant satellite line-of-sight (LoS) wavefront, so the deterministic LoS phase is common across candidate ports; port diversity therefore arises from the locally scattered field. Instead of compressing the aperture correlation into a single coefficient, we retain the full Bessel correlation matrix~\cite{10753482,11175437}. Let the normalized port coordinates be
\begin{equation}
 x_p=\frac{p-1}{P-1}W_{\rm FAS},\qquad p=1,\ldots,P,
\label{eq:fas_positions}
\end{equation}
with the obvious single-port specialization for $P=1$. Define $\mathbf R_{F}\in\mathbb C^{P\times P}$ by
\begin{equation}
[\mathbf R_F]_{p,q}=J_0\!\left(2\pi|x_p-x_q|\right),\qquad p,q=1,\ldots,P,
\label{eq:fas_corr_matrix}
\end{equation}
where $J_0(\cdot)$ is the zeroth-order Bessel function of the first kind. For each transmit-aperture index $n$, collect the diffuse coefficients across the FAS ports as
\begin{equation}
\begin{aligned}
\widetilde{\mathbf h}_{k,n}
&\triangleq[\widetilde h_{k,1,n},\ldots,\widetilde h_{k,P,n}]^T\\
&=\mathbf R_F^{1/2}\mathbf z_{k,n},\qquad
\mathbf z_{k,n}\sim\mathcal{CN}(\mathbf0,\mathbf I_P).
\end{aligned}
\label{eq:fas_channel_corr}
\end{equation}
with $\{\mathbf z_{k,n}\}$ independent over $k$ and $n$. Hence each port retains unit diffuse variance, while the correlation between any two candidate positions depends explicitly on their physical separation through \eqref{eq:fas_corr_matrix}. This full-matrix model is used in both the selected-port analysis and Monte Carlo validation. Let $\rho_k$ denote user $k$'s horizontal distance from the sub-satellite point. The large-scale gain is the free-space path loss over the corresponding slant range $R_k = \sqrt{h_s^2+\rho_k^2}$:
\begin{equation}
\beta_k = \left(\frac{\lambda}{4\pi R_k}\right)^2.
\end{equation}
Here $\mathbf h_{k,p}=[h_{k,p,1},\ldots,h_{k,p,N_h}]^T$ denotes the unit-large-scale-gain Rician vector and $\widetilde{\mathbf h}_{k,p}=[\tilde h_{k,p,1},\ldots,\tilde h_{k,p,N_h}]^T$ its scattered component; the physical communication channel is $\sqrt{\beta_k}\mathbf h_{k,p}$.

\subsubsection{Satellite-to-Target Channel (Sensing, Long Leg)}

The satellite-to-target channel is likewise modeled as Rician (LoS-dominant illumination path with a residual scattered component), with Rician factor $\kappa_t$:
\begin{align}
\mathbf g_{S,m}&=\sqrt{\beta_{S,m}}\left(\sqrt{\frac{\kappa_t}{\kappa_t+1}}\mathbf a_t(\theta_m)\right.\nonumber\\[-1mm]
&\left.\hspace{4em}+\sqrt{\frac{1}{\kappa_t+1}}\tilde{\mathbf g}_{S,m}\right),\label{eq:rician_target}\\
\tilde{\mathbf g}_{S,m}&\sim\mathcal{CN}(\mathbf0,\mathbf I_{N_h}),\nonumber
\end{align}
where $\rho_m$ denotes the target's horizontal distance from the SSP, $R_{S,m}=\sqrt{h_s^2+\rho_m^2}$ is the satellite--target slant range, $\beta_{S,m}=(\lambda/4\pi R_{S,m})^2$, and $\theta_m$ is the target departure direction relative to the satellite.

\subsubsection{Target-to-Ground-BS Channel (Sensing, Short Leg)}

The target-to-BS link has Rician factor $\kappa_b$. The reflected echo propagates from target $m$ to its associated (nearest) ground BS $b^\star(m) = \arg\min_{b\in\Phi_b}\|\mathbf{x}_b - \mathbf{x}_m\|$ over a much shorter terrestrial-scale range $r_{m,b^\star}$:
\begin{align}
\mathbf h_{m,b^\star}&=\sqrt{\beta_{m,b^\star}}\left(\sqrt{\frac{\kappa_b}{\kappa_b+1}}\mathbf a_r(\phi_m)\right.\nonumber\\[-1mm]
&\left.\hspace{4em}+\sqrt{\frac{1}{\kappa_b+1}}\tilde{\mathbf h}_{m,b^\star}\right),\label{eq:rician_bs}\\
\tilde{\mathbf h}_{m,b^\star}&\sim\mathcal{CN}(\mathbf0,\mathbf I_{N_r}),\nonumber
\end{align}
where $\mathbf{a}_r(\phi_m)\in\mathbb{C}^{N_r\times1}$ is the ground BS's conventional (non-holographic) receive steering vector toward angle-of-arrival $\phi_m$, normalized such that $\|\mathbf a_r(\phi_m)\|_2^2=N_r$, and $\beta_{m,b^\star} = (\lambda/4\pi r_{m,b^\star})^2$. Only the satellite transmit side is holographic; the ground BS receive array is a conventional multi-antenna array so no holographic amplitude-only restriction applies to the receive beamformer $\mathbf{w}_r$ defined in Section~\ref{sec:sensing}.

Because a planar PPP permits an arbitrarily small nearest-neighbor distance, directly using free-space loss would both extrapolate the far-field model into an unphysical near-field regime and make $\mathbb E[r^{-2}]$ divergent. We therefore use the regularized short-leg distance
$\hat r_{m,b^\star}=\sqrt{r_{\min}^2+d_b^2}$ with $r_{\min}>0$, where $d_b$ is the planar distance to the nearest sensing BS. This physically interpretable guard distance makes the short-hop average finite and is used in the sensing analysis below.

\subsection{Communication System (RSMA)}
\label{sec:precoding}

Let $s_c$ denote the RSMA common stream, which jointly encodes the common parts of the $K$ users' messages and also serves as the sensing-illumination waveform, and let $s_k$ denote user $k$'s private stream. The streams are mutually independent, zero mean, and normalized such that $\mathbb E[|s_c|^2]=\mathbb E[|s_k|^2]=1$. The nominal common- and private-stream powers are $P_c\triangleq\beta P_o$ and $P_k\triangleq(1-\beta)P_o/K$, respectively, where $\beta\in[0,1]$ is the RSMA common-power fraction and $P_o$ is the satellite's total nominal transmit-power budget, so that $P_c+\sum_{k=1}^{K}P_k=P_o$.

\paragraph{Angular (ephemeris-based) reference beams}
To make the ergodic performance analysis of Section~\ref{sec:analysis} tractable in closed form---and because slowly varying angular/ephemeris information is substantially easier to maintain from orbit than instantaneous full-dimensional channel-state information at the transmitter (CSIT)---we first construct deterministic \emph{reference} beams from the known directions $\{\theta_k\}_{k=1}^K$ and $\theta_m$:
\begin{align}
\mathbf w_c^{\rm ref}&=\frac{\alpha\sum_{k=1}^K\mathbf a_t(\theta_k)+\tau\mathbf a_t(\theta_m)}{\left\|\alpha\sum_{k=1}^K\mathbf a_t(\theta_k)+\tau\mathbf a_t(\theta_m)\right\|},\label{eq:wc_ideal}\\
\alpha&=\sqrt{1-\tau^2},\qquad 0\le\tau\le1,\nonumber\\
\widetilde{\mathbf W}^{\rm zf}&=\bar{\mathbf A}(\bar{\mathbf A}^H\bar{\mathbf A})^{-1},\nonumber\\
\mathbf w_k^{\rm ref}&=\widetilde{\mathbf w}^{\rm zf}_k/\|\widetilde{\mathbf w}^{\rm zf}_k\|.\label{eq:wk_ideal}
\end{align}
Here $\bar{\mathbf A}=[\mathbf a_t(\theta_1),\ldots,\mathbf a_t(\theta_K)]$ and $\widetilde{\mathbf w}^{\rm zf}_k$ is the $k$th column of $\widetilde{\mathbf W}^{\rm zf}$. Hence $\mathbf a_t(\theta_j)^H\mathbf w_k^{\rm ref}=0$ for $j\ne k$ and $\mathbf a_t(\theta_k)^H\mathbf w_k^{\rm ref}=\sqrt{g_k}$ with $g_k=1/[(\bar{\mathbf A}^H\bar{\mathbf A})^{-1}]_{kk}$.

\paragraph{Shared-state holographic realization}
The common and private reference beams are \emph{not} mapped independently to different meta-element amplitude patterns. Instead, a single deterministic shared state is formed from the aggregate angular reference
\begin{equation}
\mathbf v_{\rm agg}=\sqrt{\beta}\,\mathbf w_c^{\rm ref}+\sqrt{\frac{1-\beta}{K}}\sum_{k=1}^{K}\mathbf w_k^{\rm ref}.
\label{eq:aggregate_reference}
\end{equation}
Let $f_0$ denote a designated fixed reference feed and set $\boldsymbol\psi_{\rm ref}=\boldsymbol\psi_{f_0}$. The shared amplitudes are obtained by the deterministic holographic-recording rule
\begin{equation}
\nu_n=\left[\operatorname{Re}\!\left\{[\mathbf v_{\rm agg}]_n[\boldsymbol\psi_{\rm ref}]_n^*\right\}\right]_0^{A_{\max}},\quad n=1,\ldots,N_h,
\label{eq:shared_recording}
\end{equation}
where $[\cdot]_0^{A_{\max}}$ denotes clipping to $[0,A_{\max}]$. This produces one common $\mathbf F(\boldsymbol\nu)$ for all streams.

For $i\in\{c,1,\ldots,K\}$, let $\mathbf w_i^{\rm ref}$ denote the corresponding reference direction and choose the feed-domain vector by least-squares projection,
\begin{equation}
\widehat{\mathbf b}_i=\mathbf F(\boldsymbol\nu)^\dagger\mathbf w_i^{\rm ref},\qquad
\mathbf b_i=\xi_i\widehat{\mathbf b}_i,
\label{eq:feed_projection}
\end{equation}
where
\begin{equation}
\xi_i=\min\!\left\{1,\frac{1}{\|\widehat{\mathbf b}_i\|_2},\frac{1}{\|\mathbf F(\boldsymbol\nu)\widehat{\mathbf b}_i\|_2}\right\}
\label{eq:stream_scale}
\end{equation}
ensures $\|\mathbf b_i\|_2\le1$ and $\|\mathbf F(\boldsymbol\nu)\mathbf b_i\|_2\le1$. The realized radiating column
\begin{equation}
\mathbf q_i=\mathbf F(\boldsymbol\nu)\mathbf b_i
\label{eq:realized_column}
\end{equation}
has $\|\mathbf q_i\|_2\le1$. The stream-dependent factor $\xi_i$ is purely digital and therefore does not violate the common analog state. Define
\begin{equation}
a_i\triangleq\|\mathbf q_i\|_2,\qquad \bar{\mathbf w}_i\triangleq\mathbf q_i/a_i,
\label{eq:realized_direction}
\end{equation}
for nonzero realized columns; an unradiatable zero column is excluded from the admissible scheduled set.

\begin{remark}[Nominal, feed-domain, and radiated power]\label{rem:amplitude_power}
The nominal stream powers satisfy $P_c+\sum_kP_k=P_o$. Since $\|\mathbf b_i\|_2\le1$, the aggregate feed-domain excitation satisfies $\sum_i P_i\|\mathbf b_i\|_2^2\le P_o$. The actual radiated stream powers are $\widetilde P_i=P_i a_i^2\le P_i$. By \eqref{eq:2d_power_norm}, $a_i^2$ is also the norm squared of the complete $N_x\times N_y$ radiating column, so $\sum_i\widetilde P_i\le P_o$ is a full-aperture radiated-power statement. The analytical bounds use $\widetilde P_i$ together with the unit-norm realized directions $\bar{\mathbf w}_i$.
\end{remark}

The satellite transmit signal can therefore be written equivalently as
\begin{equation}
\begin{aligned}
\mathbf x
&=\sqrt{P_c}\,\mathbf q_c s_c
 +\sum_{k=1}^{K}\sqrt{P_k}\,\mathbf q_k s_k\\
&=\sqrt{\widetilde P_c}\,\bar{\mathbf w}_c s_c
 +\sum_{k=1}^{K}\sqrt{\widetilde P_k}\,\bar{\mathbf w}_k s_k.
\end{aligned}
\label{eq:tx_shared_rhs}
\end{equation}
Importantly, the same $\boldsymbol\nu$ and hence the same $\mathbf F(\boldsymbol\nu)$ are used for all simultaneous RSMA streams; only the feed-domain vectors $\{\mathbf b_c,\mathbf b_k\}$ differ. This explicitly enforces the shared physical RHS state rather than assuming independently realizable stream-wise holographic patterns.

For later use, define the equivalent separable full-aperture scalar projection of stream $i$ at user $k$, port $p$, as
\begin{align}
\chi_{k,p,i}&=\sqrt{\frac{\kappa_c}{\kappa_c+1}}M_{k,i}
+\sqrt{\frac{1}{\kappa_c+1}}U_{k,p,i},\label{eq:equiv_full_projection}\\
M_{k,i}&\triangleq\sqrt{N_y}\,\mathbf a_t^H(\theta_k)\bar{\mathbf w}_i,\nonumber
\end{align}
where $\mathbf U_{k,i}=[U_{k,1,i},\ldots,U_{k,P,i}]^T\sim\mathcal{CN}(\mathbf0,\mathbf R_F)$. Thus $|M_{k,i}|^2$ is exactly the two-dimensional LoS directional gain in \eqref{eq:upa_gain}, whereas the diffuse projected variance remains unity under the adopted separable scattering model.

\paragraph{FAS port selection} Each user $k$ selects, among its $P$ spatially correlated FAS ports, the port maximizing its \emph{desired private-stream} signal power,
\begin{equation}
p_k^\star \triangleq \arg\max_{p\in\{1,\ldots,P\}} |\chi_{k,p,k}|^2,
\label{eq:fas_port}
\end{equation}
This causally realizable rule depends only on the private precoder and the instantaneous port channels. We use port $p=1$ as the reference-port baseline, corresponding to reception without spatial port selection. Section~\ref{sec:analysis} first gives reference-port bounds as a baseline and then derives FAS-selected private- and common-rate bounds at the same selected physical port, including the statistical coupling induced by \eqref{eq:fas_port}.

\paragraph{Communication SINR and rates} At the selected port $p_k^\star$ in \eqref{eq:fas_port}, let $\sigma_k^2$ denote the receiver-noise power at user $k$. With no inter-satellite or inter-cell interference term, the signal-to-interference-plus-noise ratios (SINRs) expressed through \eqref{eq:equiv_full_projection} are
\begin{align}
\gamma_{k,p^\star}^{c}&=\frac{\widetilde P_c\beta_k|\chi_{k,p^\star,c}|^2}{\beta_k\sum_{i=1}^K \widetilde P_i|\chi_{k,p^\star,i}|^2+\sigma_k^2},\label{eq:gamma_common}\\
\gamma_{k,p^\star}^{p}&=\frac{\widetilde P_k\beta_k|\chi_{k,p^\star,k}|^2}{\beta_k\sum_{i\ne k}\widetilde P_i|\chi_{k,p^\star,i}|^2+\sigma_k^2},\label{eq:gamma_rates}\\
R_{c,k}&=\log_2(1+\gamma_{k,p^\star}^c),\quad R_k^p=\log_2(1+\gamma_{k,p^\star}^p).\label{eq:instant_rates}
\end{align}
For ergodic coding, the supported common throughput is $\bar R_c=\min_k\mathbb E[R_{c,k}]$ and the ergodic sum spectral efficiency is $\bar R_{\rm sum}=\bar R_c+\sum_k\mathbb E[R_k^p]$. The corresponding analytical lower bounds are derived in Section~\ref{sec:analysis}. In the numerical study, $\sigma_k^2=\sigma_b^2=k_{\rm B}T_{\rm sys}BL_{\rm ex}$, where $L_{\rm ex}$ is used in linear scale.

\paragraph{Inner-product convention} All communication and sensing projections use the conjugated inner product $\mathbf h^H\mathbf w$. This convention is required for coherent ZF precoding and matched receive combining with the generally complex ground-BS beamformer.

\subsection{User Scheduling}
\label{sec:scheduling}

The angular precoders require $K$ scheduled directions from the candidate pool. Since mapped beams need not remain orthogonal, we use a farthest-point angular scheduler: seed with the candidate nearest the target direction and repeatedly add the candidate maximizing its minimum angular gap to the selected set,
\begin{align}
\mathcal S_1&=\{\arg\min_j|\theta_j-\theta_m|\},\nonumber\\
\mathcal S_{i+1}&=\mathcal S_i\cup\left\{\arg\max_{j\notin\mathcal S_i}\min_{s\in\mathcal S_i}|\theta_j-\theta_s|\right\}.
\label{eq:scheduling}
\end{align}
The complexity is $O(N_{\rm cand}K)$.

\begin{lemma}[Scheduler-induced lower bound on the ZF array gain]\label{lem:sched_gain}
Let $\mathcal S$ be the $K$-user set produced by \eqref{eq:scheduling} (or any $K$-set), with achieved minimum pairwise angular separation $\Delta_{\min}\triangleq\min_{i\ne j\in\mathcal S}|\theta_i-\theta_j|$, all directions within $\theta_{\rm sch}\triangleq\max_{k\in\mathcal S}|\theta_k|$ of boresight. All angular quantities in the analytical inequalities of this lemma are expressed in radians. Suppose
\begin{equation}
\cos(\theta_{\rm sch})\,\Delta_{\min}\le1 \qquad\text{and}\qquad 2\sin(\theta_{\rm sch})\le1.
\label{eq:sched_gain_domain}
\end{equation}
Then the deterministic ZF array gain $g_k=1/[(\bar{\mathbf A}^H\bar{\mathbf A})^{-1}]_{kk}$ of \eqref{eq:wk_ideal} satisfies, for every $k\in\mathcal S$,
\begin{equation}
g_k \;\ge\; N_h - \frac{K-1}{\cos(\theta_{\rm sch})\,\Delta_{\min}}.
\label{eq:sched_gain_bound}
\end{equation}
In particular, $N_h > (K-1)/(\cos(\theta_{\rm sch})\Delta_{\min})$ is sufficient for a strictly positive gain guarantee across the whole scheduled set. The second domain condition in \eqref{eq:sched_gain_domain}, equivalently $\theta_{\rm sch}\le30^\circ$, ensures that the sine monotonicity step used in the proof is valid over the complete angular interval.
\end{lemma}
\begin{proof}
For half-wavelength spacing, the Dirichlet kernel gives $|\mathbf a_t(\theta_i)^H\mathbf a_t(\theta_j)|\le1/|\sin[\pi(\sin\theta_j-\sin\theta_i)/2]|$. Under \eqref{eq:sched_gain_domain}, the mean-value theorem and Jordan's inequality yield $|\mathbf a_t(\theta_i)^H\mathbf a_t(\theta_j)|\le[\cos(\theta_{\rm sch})\Delta_{\min}]^{-1}$. The Gram matrix $\mathbf G=\bar{\mathbf A}^H\bar{\mathbf A}$ therefore has diagonal $N_h$ and off-diagonal magnitudes bounded by this quantity. Gershgorin's theorem gives $\lambda_{\min}(\mathbf G)\ge N_h-(K-1)/[\cos(\theta_{\rm sch})\Delta_{\min}]$. Since $[\mathbf G^{-1}]_{kk}\le1/\lambda_{\min}(\mathbf G)$, \eqref{eq:sched_gain_bound} follows.
\end{proof}

\paragraph{Scheduler quality and scope}
The rule in \eqref{eq:scheduling} follows the farthest-point traversal principle: after the target-nearest seed is fixed, each subsequent user maximizes its minimum angular distance from the already selected set. We use it as a low-complexity geometry-conditioning heuristic and compare it directly with random and exhaustive max-separation scheduling in Section~\ref{sec:numerical_results}; no approximation factor for the pairwise-separation objective is required by Lemma~\ref{lem:sched_gain}. The lemma bounds only the pre-mapping ZF gain; post-mapping leakage still depends on the full angle geometry and reference-wave phase profile.

\subsection{Sensing System (Bistatic)}
\label{sec:sensing}

The common stream $s_c$ illuminates target $m$; its echo, scaled by the dimensionless reflection factor $\zeta_m$ associated with physical RCS $\sigma_m$, propagates to the target's nearest ground BS $b^\star$, which applies a conventional unit-norm receive beamformer $\mathbf w_r$ matched to the target angle of arrival $\phi_m$. Let $\mathbf n_b\sim\mathcal{CN}(\mathbf0,\sigma_b^2\mathbf I_{N_r})$ denote the sensing-BS receiver noise. With bandwidth $B$ and coherent sensing interval $T_s$, the matched-filter gain is $G_p\triangleq BT_s$.

Consistent with the separable full-aperture model, define
\begin{align}
\chi_{m,c}&=\sqrt{\frac{\kappa_t}{\kappa_t+1}}M_{m,c}+\sqrt{\frac{1}{\kappa_t+1}}U_{m,c},\\
M_{m,c}&\triangleq\sqrt{N_y}\,\mathbf a_t^H(\theta_m)\bar{\mathbf w}_c,\qquad U_{m,c}\sim\mathcal{CN}(0,1).
\label{eq:target_projection}
\end{align}
The sensing BS is synchronized to the known common waveform through ephemeris-assisted timing/frequency compensation and network-side waveform knowledge. Neglecting residual direct-path leakage and the vanishing cross-correlation of independently encoded private streams in the considered range--Doppler cell, the sensing metric is noise limited:
\begin{equation}
\gamma_m^{s} = \frac{\widetilde P_c\,\zeta_m\, G_p}{\sigma_b^2}
\left|\mathbf w_r^H \mathbf h_{m,b^\star}\right|^2\left|\chi_{m,c}\right|^2.
\label{eq:sensing_sinr}
\end{equation}
Expanding via \eqref{eq:rician_target}--\eqref{eq:rician_bs},
$\gamma_m^s$ is proportional to the product of two independent
Rician quadratic-form gains and the two-leg path loss
$\beta_{S,m}\beta_{m,b^\star}$. The short-leg distance
$r_{m,b^\star}$ follows the nearest-neighbor distribution induced by the sensing-BS PPP $\Phi_b$. Let $\Gamma_s>0$ denote the prescribed sensing-SNR threshold; a target is declared detected when $\gamma_m^s \ge \Gamma_s$. The closed-form average $\bar\gamma^s = \mathbb{E}[\gamma_m^s]$ is derived in Theorem~\ref{thm:sensing}.

\begin{remark}[RCS normalization]\label{rem:rcs_normalization}
The two Friis one-way gains in \eqref{eq:sensing_sinr} contribute $\lambda^4/(4\pi)^4$, whereas the standard bistatic radar equation contains $\lambda^2\sigma_m/(4\pi)^3$. Dimensional consistency therefore requires the dimensionless reflection factor
\begin{equation}
\zeta_m=\frac{4\pi}{\lambda^2}\sigma_m,
\label{eq:rcs_mapping}
\end{equation}
where $\sigma_m$ is the physical RCS in $\mathrm m^2$.
\end{remark}

The above resolution-cell model isolates the bistatic propagation and RHS/FAS effects targeted by this performance analysis. Residual direct-path leakage, finite-integration private-stream leakage, and explicit cancellation design are outside the present scope.

\section{Ergodic Performance Analysis}
\label{sec:analysis}

This section derives tractable ergodic communication bounds and the average bistatic sensing SNR under the deterministic angular precoder design of Section~\ref{sec:precoding}. Theorems~\ref{thm:private} and~\ref{thm:common} provide reference-port baselines. Lemma~\ref{lem:fas} and Theorem~\ref{thm:fas_private} then account for best-of-$P$ FAS selection in the private layer, including its induced interference coupling. Lemma~\ref{lem:fas_crosslog} and Theorem~\ref{thm:fas_common} extend the same selected-port treatment to common-stream decoding, so both SIC stages are evaluated at the physical port chosen by \eqref{eq:fas_port}. The fixed-range expressions are extended to angle-conditioned spatial averaging in Section~\ref{sec:spatial_avg}.

\subsection{Preliminaries}

\begin{lemma}[Rician quadratic form against a fixed vector]\label{lem:rician}
Let $\mathbf{h} = \sqrt{\frac{\kappa}{\kappa+1}}\bar{\mathbf{a}} + \sqrt{\frac{1}{\kappa+1}}\tilde{\mathbf{h}}$ with $\tilde{\mathbf{h}}\sim\mathcal{CN}(\mathbf{0},\mathbf{I}_N)$, and let $\mathbf{w}$ be any \emph{deterministic}, unit-norm vector. Then $\mathbf{h}^H\mathbf{w} \sim \mathcal{CN}\!\left(\sqrt{\tfrac{\kappa}{\kappa+1}}\bar{\mathbf{a}}^H\mathbf{w},\ \tfrac{1}{\kappa+1}\right)$ (by rotational invariance of the isotropic complex Gaussian $\tilde{\mathbf h}$), so that $G=|\mathbf{h}^H\mathbf{w}|^2$ is (a scaled) noncentral chi-square with 2 degrees of freedom, with
\begin{align}
\mathbb E[G]&=\frac{\kappa g_w+1}{\kappa+1},\nonumber\\
\mathbb E[\ln G]&=-\ln(\kappa+1)+S(\kappa g_w),\nonumber\\
g_w&\triangleq|\bar{\mathbf a}^H\mathbf w|^2.
\label{eq:lem1}
\end{align}
where $S(\lambda_0)\triangleq \sum_{j=0}^{\infty} e^{-\lambda_0}\lambda_0^j\psi(j+1)/j!$, and $\psi(\cdot)$ denotes the digamma function. This series follows from the Poisson-mixture representation of the noncentral chi-square distribution.
\end{lemma}

\begin{lemma}[Exact direction-aware shared-RHS realized gain]\label{lem:eta}
For any deterministic unit-norm principal-plane realized direction $\bar{\mathbf w}_i$ in \eqref{eq:realized_direction}, the corresponding full $N_x\times N_y$ directional power gain toward deterministic angle $\theta$ is
\begin{equation}
g(\theta;i)=N_y\left|\mathbf a_t^H(\theta)\bar{\mathbf w}_i\right|^2
=\left|\mathbf a_{2\rm D}^H(\theta)\bar{\mathbf w}_i^{2\rm D}\right|^2.
\label{eq:holo_gain}
\end{equation}
Because $\boldsymbol\nu$ and the feed-domain projections are deterministic functions of the scheduled/target angles, these gains are exactly computable finite-dimensional quantities. We use
\begin{align}
g_{k,k}&=g(\theta_k;k),&
 g_{k,i}&=g(\theta_k;i),\nonumber\\
g_{c,k}&=g(\theta_k;c),&
 g_{c,m}&=g(\theta_m;c).
\label{eq:gain_defs}
\end{align}
The realized-column norms are accounted for separately through $\widetilde P_i=P_i a_i^2$.
\end{lemma}

\paragraph{Scalar efficiency baseline}
A geometry-independent proxy $\eta=1/4$ ($-6.02$~dB) is retained only as a simple reference. For this benchmark, each unconstrained unit-norm reference direction is preserved and only its radiated amplitude is reduced,
\begin{equation}
\mathbf w_i^{\rm scalar}=\sqrt\eta\,\mathbf w_i^{\rm ref},\qquad i\in\{c,1,\ldots,K\},
\end{equation}
so that $\widetilde P_i^{\rm scalar}=\eta P_i$. This proxy cannot represent the direction-dependent self- and cross-gains created when all streams are projected through the same amplitude-constrained multi-feed transfer matrix; the deployed model therefore uses the exact deterministic gains in \eqref{eq:holo_gain}--\eqref{eq:gain_defs}.

\begin{lemma}[FAS best-of-$P$ port-selection self-gain log-moment]\label{lem:fas}
Let $G_{k,p} \triangleq |\mathbf h_{k,p}^H\bar{\mathbf w}_k|^2$ and $G_{k,\max}\triangleq\max_{1\le p\le P}G_{k,p}$. Under the full Bessel correlation model \eqref{eq:fas_corr_matrix}--\eqref{eq:fas_channel_corr}, define the projected diffuse vector
\begin{equation}
\mathbf U_k\triangleq[U_{k,1},\ldots,U_{k,P}]^T
\sim\mathcal{CN}(\mathbf0,\mathbf R_F),
\label{eq:projected_fas_vector}
\end{equation}
which follows because the same deterministic unit-norm beam $\bar{\mathbf w}_k$ is projected onto every port and the diffuse coefficients are independent across transmit-aperture indices. Writing $L_k=\sqrt{g_{k,k}}$ without loss of generality after absorbing the deterministic phase, the port projections are
\begin{equation}
H_{k,p}=\sqrt{\frac{\kappa_c}{\kappa_c+1}}L_k
+\sqrt{\frac{1}{\kappa_c+1}}U_{k,p},
\qquad G_{k,p}=|H_{k,p}|^2.
\label{eq:fullcorr_port_projection}
\end{equation}
Let $\mathbf z\sim\mathcal{CN}(\mathbf0,\mathbf I_P)$ and $\mathbf U_k=\mathbf R_F^{1/2}\mathbf z$. Define
\begin{equation}
\mathcal G_k(\mathbf z)\triangleq
\max_{1\le p\le P}\left|
\sqrt{\frac{\kappa_c}{\kappa_c+1}}L_k+
\frac{[\mathbf R_F^{1/2}\mathbf z]_p}{\sqrt{\kappa_c+1}}
\right|^2.
\label{eq:fullcorr_gmax_integrand}
\end{equation}
Then the selected-self log moment admits the exact $P$-dimensional complex-Gaussian integral
\begin{equation}
\mathbb E[\ln G_{k,\max}]
=\frac{1}{\pi^P}\int_{\mathbb C^P}
 e^{-\|\mathbf z\|_2^2}\ln \mathcal G_k(\mathbf z)\,d\mathbf z.
\label{eq:fas_logmoment}
\end{equation}
Moreover, with
\begin{equation}
p_k^\star=\arg\max_p G_{k,p},\qquad
A_k^\star\triangleq U_{k,p_k^\star},
\label{eq:selected_diffuse_fullcorr}
\end{equation}
the moments $\mu_{A,k}=\mathbb E[A_k^\star]$ and $\nu_{A,k}=\mathbb E[|A_k^\star|^2]$ are given by the same Gaussian integral with the integrands $[\mathbf R_F^{1/2}\mathbf z]_{p_k^\star}$ and $|[\mathbf R_F^{1/2}\mathbf z]_{p_k^\star}|^2$, respectively. These quantities depend on $(\kappa_c,g_{k,k},\mathbf R_F)$ and are evaluated numerically using deterministic Gaussian quadrature for small $P$ or low-discrepancy Gaussian quadrature/quasi-Monte-Carlo integration for larger $P$.
\end{lemma}
\begin{proof}
For each transmit-aperture index $n$, the diffuse port vector has covariance $\mathbf R_F$. Applying the same deterministic unit-norm projection $\bar{\mathbf w}_k$ across all ports preserves this covariance, giving \eqref{eq:projected_fas_vector}. Equation~\eqref{eq:fullcorr_port_projection} then follows from the Rician model. Substituting $\mathbf U_k=\mathbf R_F^{1/2}\mathbf z$ into the definitions of $G_{k,\max}$ and $A_k^\star$ and averaging with respect to the standard complex-Gaussian density $\pi^{-P}e^{-\|\mathbf z\|_2^2}$ gives the stated integral representations directly.
\end{proof}

\begin{remark}[Selection-induced interference coupling]\label{rem:fas_coupling}
Lemma~\ref{lem:fas} characterizes the selected self-gain $G_{k,\max}$ only. In general, the interference gain $G_{k,i}(p_k^\star)$ observed at the same selected port is \emph{not} distributed as the interference gain at an arbitrary reference port because both projections are generated by the same scattered channel vector. Independence would hold for orthogonal mapped beam directions, but the shared-RHS projection does not generally preserve the orthogonality of the ideal ZF beams. Therefore, directly combining the best-port numerator statistic with an unchanged reference-port interference mean is not generally justified. Theorem~\ref{thm:fas_private} resolves this coupling by explicitly retaining the post-projection beam correlation.
\end{remark}

\subsection{Communication Rates}

\begin{theorem}[Ergodic private rate, conservative lower bound]\label{thm:private}
Under the angular precoder design, Lemma~\ref{lem:rician}, and the port-1/no-inter-cell-interference setting,
\begin{align}
D_k^{p}&\triangleq \beta_k\sum_{i\ne k}\tilde P_i
\frac{\kappa_c g_{k,i}+1}{\kappa_c+1}+\sigma_k^2,\nonumber\\
R_k^{p,\rm Lb}&=\log_2\!\left(1+
\frac{\tilde P_k\beta_k e^{S(\kappa_c g_{k,k})}}
{(\kappa_c+1)D_k^p}\right),\nonumber\\
\bar R_k^p&\ge R_k^{p,\rm Lb}.
\label{eq:thm_private}
\end{align}
with $g_{k,k}, g_{k,i}$ the exact shared-RHS realized gains of Lemma~\ref{lem:eta} \eqref{eq:gain_defs}, and $\tilde P_i = P_i a_i^2$ the effective radiated stream powers of Remark~\ref{rem:amplitude_power}.
\end{theorem}
\begin{proof}
With $X=\widetilde P_k\beta_kG_k^{\rm self}$ and $Y=\beta_k\sum_{i\ne k}\widetilde P_iG_{k,i}^{\rm int}+\sigma_k^2$, define $Z=\ln X-\ln Y$. Convexity of $\log_2(1+e^Z)$ gives the first Jensen step, while $\mathbb E[\ln Y]\le\ln\mathbb E[Y]$ gives the second; together they yield
\begin{equation}
\mathbb E[\log_2(1+X/Y)]\ge\log_2\!\left(1+\frac{e^{\mathbb E[\ln X]}}{\mathbb E[Y]}\right),
\end{equation}
without any independence requirement. Substitution of Lemma~\ref{lem:rician}'s log moment and mean power yields \eqref{eq:thm_private}.
\end{proof}

\begin{theorem}[Reference-port ergodic common-rate lower bound]\label{thm:common}
Define the supported ergodic RSMA common rate as
\begin{equation}
\bar R_c \triangleq \min_{1\le k\le K}\mathbb E\!\left[\log_2(1+\gamma^c_{k,1})\right],
\label{eq:ergodic_common_def}
\end{equation}
i.e., the common codeword is transmitted across fading states at a rate decodable by every scheduled user. Under the reference-port setting,
\begin{align}
D_k^{c}&\triangleq \beta_k\sum_{i=1}^{K}\tilde P_i
\frac{\kappa_c g_{k,i}+1}{\kappa_c+1}+\sigma_k^2,\nonumber\\
R_k^{c,\mathrm{Lb}}&=\log_2\!\left(1+
\frac{\tilde P_c\beta_k e^{S(\kappa_c g_{c,k})}}
{(\kappa_c+1)D_k^c}\right),\label{eq:thm_common}\\
\bar R_c&\ge \min_{1\le k\le K}R_k^{c,\mathrm{Lb}}.\label{eq:thm_common_min}
\end{align}
\end{theorem}
\begin{proof}
Apply the Jensen argument used in Theorem~\ref{thm:private} separately to each user's common-decoding SINR. Since \eqref{eq:ergodic_common_def} takes the minimum after the per-user ergodic expectations, the minimum of valid per-user lower bounds remains a valid lower bound.
\end{proof}

\begin{theorem}[FAS-enabled ergodic private rate lower bound]\label{thm:fas_private}
Let $\rho_{k,i}\triangleq(\bar{\mathbf w}_i)^H\bar{\mathbf w}_k$ be the (complex, exactly computable) post-projection precoder cross-correlation, and $M_{k,i}\triangleq\sqrt{N_y}\,\mathbf{a}_t^H(\theta_k)\bar{\mathbf w}_i$ the full-aperture LoS projection underlying $g_{k,i}=|M_{k,i}|^2$ (Lemma~\ref{lem:eta}). Define the self-gain-selection moments
\begin{equation}
A_k^\star\triangleq U_{p_k^\star},\qquad \mu_{A,k} \triangleq \mathbb{E}[A_k^\star], \qquad \nu_{A,k} \triangleq \mathbb{E}\big[|A_k^\star|^2\big],
\label{eq:muA_sigmaA}
\end{equation}
the first and second moments of the scattered-channel self-projection $U_p$ (Lemma~\ref{lem:fas}) evaluated at the FAS-selected port $p_k^\star$ \eqref{eq:fas_port} -- functions of $(\kappa_c,g_{k,k},\mathbf R_F)$ \emph{only}, independent of any interferer $i$, hence computed once per user $k$ and reused across all $i$. Then, under the angular precoder design and the same Jensen-bound recipe as Theorem~\ref{thm:private}, requiring no independence between numerator and denominator,
\begin{align}
\bar R_k^{p,\mathrm{FAS}}
&\ge\log_2\!\left(1+
\frac{\tilde P_k\beta_k e^{\mathbb E[\ln G_{k,\max}]}}
{\beta_k\sum_{i\ne k}\tilde P_i\Xi_{k,i}+\sigma_k^2}\right),\nonumber\\
\Xi_{k,i}&\triangleq\frac{\kappa_c g_{k,i}+1}{\kappa_c+1}
+\frac{2\sqrt{\kappa_c}}{\kappa_c+1}
\Re\{M_{k,i}^*\rho_{k,i}^*\mu_{A,k}\}\nonumber\\
&\quad+\frac{|\rho_{k,i}|^2(\nu_{A,k}-1)}{\kappa_c+1}.
\label{eq:thm4}
\end{align}
with $\mathbb{E}[\ln G_{k,\max}]$ from Lemma~\ref{lem:fas}. $\Xi_{k,i}$ reduces exactly to Lemma~\ref{lem:rician}'s port-invariant $(\kappa_c g_{k,i}+1)/(\kappa_c+1)$ when $\rho_{k,i}=0$.
\end{theorem}
\begin{proof}
For the selected self-beam projection $U_p$ and any mapped interferer, joint Gaussianity gives the regression decomposition $Y_p=\rho_{k,i}^*U_p+\sqrt{1-|\rho_{k,i}|^2}W_p$, where $\mathbf W=[W_1,\ldots,W_P]^T\sim\mathcal{CN}(\mathbf0,\mathbf R_F)$ is an independent copy of $\mathbf U_k$ and is independent of the complete self-projection vector $\mathbf U_k$. Since $p_k^\star$ depends only on $\{U_p\}$, $W_{p_k^\star}$ remains zero mean with unit second moment. Hence
\begin{equation}
\mathbb E[Y_{p_k^\star}]=\rho_{k,i}^*\mu_{A,k},\qquad \mathbb E|Y_{p_k^\star}|^2=1+|\rho_{k,i}|^2(\nu_{A,k}-1).
\end{equation}
Expanding the Rician interference projection gives $\Xi_{k,i}$ in \eqref{eq:thm4}. Applying the Jensen step of Theorem~\ref{thm:private} with Lemma~\ref{lem:fas}'s selected-self log moment proves the bound. The moments $\mu_{A,k},\nu_{A,k}$ are the full-correlation selected-port moments defined in Lemma~\ref{lem:fas} and are evaluated from the same $\mathbf R_F$-dependent Gaussian integral.
\end{proof}

\begin{lemma}[Selected-port cross-beam log moment]\label{lem:fas_crosslog}
Let $A_k^\star$ be the selected scattered self-beam projection defined in Theorem~\ref{thm:fas_private}. For any deterministic beam $j$ (including the common beam), define $\rho_{k,j}=(\bar{\mathbf w}_j)^H\bar{\mathbf w}_k$ and $M_{k,j}=\sqrt{N_y}\,\mathbf a_t^H(\theta_k)\bar{\mathbf w}_j$. The selected scattered projection satisfies
\begin{equation}
Y_{k,j}^\star=\rho_{k,j}^*A_k^\star+\sqrt{1-|\rho_{k,j}|^2}W,\qquad W\sim\mathcal{CN}(0,1),
\end{equation}
with $W$ independent of the selection process. Therefore, conditioned on $A_k^\star=a$,
\begin{align}
H_{k,j}^\star|a&\sim\mathcal{CN}(m_{k,j}(a),v_{k,j}),\\
m_{k,j}(a)&=\sqrt{\frac{\kappa_c}{\kappa_c+1}}M_{k,j}+\frac{\rho_{k,j}^*a}{\sqrt{\kappa_c+1}},\\
v_{k,j}&=\frac{1-|\rho_{k,j}|^2}{\kappa_c+1}.
\end{align}
Thus
\begin{equation}
L_{k,j}\triangleq\mathbb E[\ln|H_{k,j}^\star|^2]=\mathbb E_{A_k^\star}\!\left[\ln v_{k,j}+S\!\left(\frac{|m_{k,j}(A_k^\star)|^2}{v_{k,j}}\right)\right],
\label{eq:fas_cross_log}
\end{equation}
with the continuous limit $\ln|m_{k,j}|^2$ for $v_{k,j}=0$. The outer expectation is evaluated using the $\mathbf R_F$-dependent selected-variable law in Lemma~\ref{lem:fas}.
\end{lemma}

\begin{remark}[Effect of the full Bessel correlation model]\label{rem:fullcorr_structure}
Replacing the one-factor approximation by the full matrix $\mathbf R_F$ changes the joint law of the port vector and therefore the numerical values of $\mathbb E[\ln G_{k,\max}]$, $\mu_{A,k}$, $\nu_{A,k}$, and $\mathcal L_{k,j}$. However, the algebraic forms of Theorem~\ref{thm:fas_private}, Lemma~\ref{lem:fas_crosslog}, and Theorem~\ref{thm:fas_common} remain unchanged. The reason is that for any other deterministic beam $i$, the full projected diffuse vectors satisfy the vector Gaussian regression identity
\begin{equation}
\mathbf Y_{k,i}=\rho_{k,i}^*\mathbf U_k+\sqrt{1-|\rho_{k,i}|^2}\,\mathbf W_{k,i},
\end{equation}
where $\mathbf W_{k,i}\sim\mathcal{CN}(\mathbf0,\mathbf R_F)$ is independent of $\mathbf U_k$. Since $p_k^\star$ is a function only of $\mathbf U_k$, the selected residual $W_{k,i,p_k^\star}$ remains zero mean with unit variance and independent of the selection process. Hence the selected-port interference mean and cross-beam log-moment derivations retain the same scalar forms once the full-correlation selected moments from Lemma~\ref{lem:fas} are substituted.
\end{remark}

\begin{theorem}[FAS-selected ergodic common-rate lower bound]\label{thm:fas_common}
Under the private-self-gain port-selection rule \eqref{eq:fas_port}, the same physical selected port is used to decode the common stream before SIC. Let $\Xi_{k,i}$ be the exact selected-port mean power in \eqref{eq:thm4}, extended to $i=k$ by the same formula, and let $\mathcal L_{k,c}$ be Lemma~\ref{lem:fas_crosslog}'s common-beam log moment. Then
\begin{align}
R_{k}^{c,\mathrm{FAS,Lb}}
&=\log_2\!\left(1+
\frac{\tilde P_c\beta_k e^{\mathcal L_{k,c}}}
{\beta_k\sum_{i=1}^{K}\tilde P_i\Xi_{k,i}+\sigma_k^2}\right),\label{eq:fas_common_lb}\\
\bar R_c&\ge\min_{1\le k\le K}R_{k}^{c,\mathrm{FAS,Lb}}.\label{eq:fas_common_min}
\end{align}
\end{theorem}
\begin{proof}
The numerator log moment follows exactly from Lemma~\ref{lem:fas_crosslog}. The denominator mean follows from the regression decomposition already used in Theorem~\ref{thm:fas_private}; this includes the selected private self-stream ($i=k$) and all cross-private streams. Applying the Jensen argument of Theorem~\ref{thm:private} yields the per-user lower bound, and the ergodic common-rate definition \eqref{eq:ergodic_common_def} makes the final minimum rigorous.
\end{proof}

\subsection{Sensing Performance}

\begin{theorem}[Average bistatic sensing SNR]\label{thm:sensing}
Let the ground-BS receive beamformer use conventional maximum-ratio combining (MRC) toward the target's angle of arrival, $\mathbf{w}_r = \mathbf{a}_r(\phi_m)/\sqrt{N_r}$ (unconstrained; no holographic efficiency loss on the receive leg). Using the regularized short-leg distance $\hat r_{m,b^\star}=\sqrt{r_{\min}^2+d_b^2}$ and the angular common precoder of Section~\ref{sec:precoding},
\begin{align}
\bar\gamma^s&=\frac{\tilde P_c\zeta_m\beta_{S,m}G_p}{\sigma_b^2}
\frac{\kappa_t g_{c,m}+1}{\kappa_t+1}
\frac{\kappa_bN_r+1}{\kappa_b+1}\nonumber\\
&\quad\times\left(\frac{\lambda}{4\pi}\right)^2
\pi\lambda_b e^{z_b}E_1(z_b),\qquad
z_b\triangleq\pi\lambda_b r_{\min}^2.
\label{eq:thm_sensing}
\end{align}
where $g_{c,m}=N_y|\mathbf a_t(\theta_m)^H\bar{\mathbf w}_c|^2$ is the full two-dimensional shared-RHS target gain of Lemma~\ref{lem:eta}, $\tilde P_c=P_ca_c^2$ is the effective radiated common-stream power (Remark~\ref{rem:amplitude_power}), and $E_1(\cdot)$ is the exponential integral function.
\end{theorem}
\begin{proof}
The transmit and receive Rician gains are independent, so their means multiply exactly. Lemma~\ref{lem:rician} gives $(\kappa_tg_{c,m}+1)/(\kappa_t+1)$ on the long leg and $(\kappa_bN_r+1)/(\kappa_b+1)$ under receive MRC. For the nearest BS, $d_b$ has density $2\pi\lambda_bd e^{-\pi\lambda_bd^2}$ and $\hat r^2=r_{\min}^2+d_b^2$, yielding
\begin{equation}
\mathbb E[\hat r^{-2}]=\pi\lambda_b e^{z_b}E_1(z_b),\qquad z_b=\pi\lambda_b r_{\min}^2.
\end{equation}
Combining these factors proves \eqref{eq:thm_sensing}.
\end{proof}

\subsection{Spatial Averaging over the Footprint}
\label{sec:spatial_avg}

Theorems~\ref{thm:private}--\ref{thm:sensing} are geometry-conditioned: they average over small-scale fading (and, for Theorem~\ref{thm:sensing}, over the nearest sensing-BS distance) for fixed satellite-to-ground slant ranges. We now average over the residual ground-position randomness induced by the finite satellite footprint. This averaging concerns the finite-window user/target processes $\Phi_u\cap\mathcal D$ and $\Phi_r\cap\mathcal D$; the sensing-BS process remains the planar PPP specified in Section~\ref{sec:system_model}.

A standard finite-window PPP property is useful here: conditioned on $N_u=|\Phi_u\cap\mathcal D|=n$, the $n$ user locations are i.i.d. uniform on $\mathcal D$ (and analogously for targets). Under the principal-plane array abstraction, the modeled departure angle depends only on the ground coordinate $x$ through $\theta=\arctan(x/h_s)$, while the orthogonal coordinate $y$ affects the slant range but not the one-dimensional steering response. Consequently, once a scheduled angle is fixed, $x=h_s\tan\theta$ is fixed and only the admissible $y$ coordinate remains random. Because the scheduler in \eqref{eq:scheduling} uses only the candidate angles, conditioning on the event that a particular $x$-coordinate is selected does not alter the conditional law of its associated $y$ coordinate.

\begin{lemma}[Unconditional footprint average of the free-space path loss]\label{lem:beta_spatial}
Let a candidate point be drawn from the finite-window PPP conditional location law in $\mathcal D$. Its horizontal radius $\rho=\sqrt{x^2+y^2}$ has density
\begin{equation}
f_\rho(\rho)=\frac{2\rho}{r_c^2},\qquad 0\le \rho\le r_c.
\end{equation}
For $\beta(\rho)=(\lambda/4\pi)^2(h_s^2+\rho^2)^{-1}$,
\begin{equation}
\mathbb E[\beta(\rho)]
=\left(\frac{\lambda}{4\pi}\right)^2\frac{1}{r_c^2}
\ln\!\left(1+\frac{r_c^2}{h_s^2}\right).
\label{eq:Ebeta}
\end{equation}
This is the average for an unscheduled candidate drawn uniformly from the footprint; it is not the conditional path-loss law after an angle has been selected.
\end{lemma}
\begin{proof}
Conditioned on the finite-window point count, the point is uniform on $\mathcal D$, hence $f_\rho(\rho)=2\rho/r_c^2$. Direct integration of $\beta(\rho)f_\rho(\rho)$ over $[0,r_c]$ gives \eqref{eq:Ebeta}.
\end{proof}

\begin{lemma}[Footprint average conditional on a scheduled angle]\label{lem:beta_conditional}
Fix an admissible scheduled angle $\theta$, so that $x=h_s\tan\theta$ and $|x|<r_c$. Define
\begin{equation}
Y(x)\triangleq\sqrt{r_c^2-x^2},\qquad
A(x)\triangleq\sqrt{h_s^2+x^2}.
\end{equation}
For a point uniformly distributed on $\mathcal D$,
\begin{equation}
y\mid x\sim\mathrm{Uniform}[-Y(x),Y(x)].
\label{eq:y_given_x}
\end{equation}
Moreover, \eqref{eq:y_given_x} remains valid for a user selected by the scheduler in \eqref{eq:scheduling}, because that scheduler is measurable with respect to the candidate $x$-coordinates (equivalently, angles) only. Therefore,
\begin{equation}
\mathbb E[\beta(\rho)\mid\theta]
=\left(\frac{\lambda}{4\pi}\right)^2
\frac{\arctan\!\big(Y(x)/A(x)\big)}{Y(x)A(x)}.
\label{eq:Ebeta_cond}
\end{equation}
The same expression applies to the representative target after replacing $\theta$ by $\theta_m$.
\end{lemma}
\begin{proof}
The joint conditional location density in the disk is constant. At fixed $x$, its support is the chord $[-Y(x),Y(x)]$, proving \eqref{eq:y_given_x}. The scheduler depends only on the set of $x$-coordinates, so conditioning additionally on the corresponding selection event does not reweight $y$ once $x$ is fixed. Finally,
\[
\frac{1}{2Y}\int_{-Y}^{Y}\frac{dy}{h_s^2+x^2+y^2}
=\frac{\arctan(Y/A)}{YA},
\]
which gives \eqref{eq:Ebeta_cond}.
\end{proof}

\begin{remark}[Mean-path-loss substitution is not a conservative rate operation]\label{rem:spatial_naive}
For fixed fading-averaged beam quantities, let $\ell>0$ denote a generic path-loss gain (introduced here only to avoid overloading the RSMA power fraction $\beta$). Each per-user communication bound in Theorems~\ref{thm:private}--\ref{thm:fas_common} can be written as
\begin{equation}
f(\ell)=\log_2\!\left(1+\frac{A\ell}{D\ell+\sigma^2}\right),
\qquad A>0,\;D\ge0,\;\sigma^2>0.
\end{equation}
Equivalently,
$f(\ell)=\log_2(\sigma^2+(A+D)\ell)-\log_2(\sigma^2+D\ell)$, and direct differentiation gives
\begin{align}
f'(\ell)
&=\frac{A\sigma^2}
{\ln2\,[\sigma^2+(A+D)\ell][\sigma^2+D\ell]}>0,\\
f''(\ell)
&=-\frac{A\sigma^2}{\ln2}
\left[
\frac{A+D}{[\sigma^2+(A+D)\ell]^2[\sigma^2+D\ell]}
\right.\nonumber\\[-0.2em]
&\hspace{7em}\left.
+\frac{D}{[\sigma^2+(A+D)\ell][\sigma^2+D\ell]^2}
\right]<0.
\label{eq:spatial_concavity}
\end{align}
Thus $f$ is increasing and concave, including the interference-free special case $D=0$. Jensen's inequality therefore yields
\begin{equation}
\mathbb E[f(\ell)\mid\theta]\le f\!\left(\mathbb E[\ell\mid\theta]\right).
\end{equation}
Hence inserting \eqref{eq:Ebeta_cond} directly into the fixed-geometry rate bound is generally an \emph{upper approximation to the spatial average of that bound}, not a conservative lower bound on the spatially averaged rate.
\end{remark}

\begin{corollary}[Valid spatially averaged communication lower bounds]\label{cor:spatial_comm}
Because each fixed-$\beta_k$ expression $R_k^{\mathrm{Lb}}(\beta_k)$ is a lower bound on the corresponding fading-averaged rate, averaging the bound itself preserves the inequality:
\begin{align}
\bar R_k^{\mathrm{sp,Lb}}
&\triangleq \mathbb E_{y\mid\theta_k}
\!\left[R_k^{\mathrm{Lb}}\big(\beta_k(x_k,y)\big)\right]\\
&\le \mathbb E_{y\mid\theta_k}
\!\left[\mathbb E_{\rm fading}[R_k\mid x_k,y]\right],
\qquad x_k=h_s\tan\theta_k.
\label{eq:cor_spatial}
\end{align}
The expectation is a one-dimensional integral over the uniform chord in \eqref{eq:y_given_x} and is evaluated by Gauss--Legendre quadrature. For the common stream, a valid supported spatial lower bound is
\begin{equation}
\bar R_c^{\mathrm{sp,Lb}}
=\min_k\;\mathbb E_{y\mid\theta_k}
\!\left[R_{c,k}^{\mathrm{Lb}}\big(\beta_k(x_k,y)\big)\right],
\label{eq:spatial_common}
\end{equation}
with either the reference-port or FAS-selected per-user bound used consistently.
\end{corollary}

\begin{corollary}[Exact spatial average of the mean sensing SNR]\label{cor:spatial_sensing}
For fixed target angle $\theta_m$, the mean sensing SNR in \eqref{eq:thm_sensing} is linear in the satellite--target path-loss factor $\beta_{S,m}$. Therefore,
\begin{equation}
\mathbb E_{y\mid\theta_m}
\!\left[\bar\gamma^s\big(\beta_{S,m}(x_m,y)\big)\right]
=\bar\gamma^s\!\left(\mathbb E[\beta_{S,m}\mid\theta_m]\right),
\label{eq:cor_spatial_sensing}
\end{equation}
where the conditional mean is given by \eqref{eq:Ebeta_cond}. This equality holds because all other factors in Theorem~\ref{thm:sensing}, including the realized target-direction gain, are fixed once $\theta_m$ and the deterministic beam design are fixed.
\end{corollary}

\begin{table}[!t]
\caption{Nominal Simulation Parameters}
\label{tab:sim_parameters}
\centering
\footnotesize
\begin{tabular}{@{}ll@{}}
\toprule
Parameter & Nominal value \\
\midrule
Carrier frequency $f_c$ & $20$~GHz \\
Satellite altitude $h_s$ & $550$~km \\
Bandwidth $B$ & $200$~MHz \\
System temperature $T_{\rm sys}$ & $300$~K \\
Receiver implementation loss $L_{\rm ex}$ & $2.5$ dB \\
RHS size $N_x\times N_y$ & $128\times128$ ($N_{\rm RHS}=16384$) \\
Element spacing; $k_g/k_0$ & $\lambda/2$; $1$ \\
RHS modulation limit $A_{\max}$ & $1$ (dimensionless) \\
Active RHS feeds $N_f$ & $25$ \\
Fixed feed direction cosines $\{\xi_f\}$ & Uniformly spaced on $[-0.5,0.5]$ \\
Scheduled user angles & $\{-28^\circ,-11^\circ,8^\circ,24^\circ\}$ \\
Target angle $\theta_m$ & $3^\circ$ \\
Nominal transmit power $P_o$ & $150$~W \\
Common-power fraction $\beta$ & $0.3$ \\
Common-beam steering $\tau$ & $0.3$ \\
Rician factors $\kappa_c,\kappa_t,\kappa_b$ & $10,10,6$~dB \\
FAS ports $P$; aperture $W_{\rm FAS}$ & $8$; $2$ wavelengths \\
Sensing-BS antennas $N_r$ & $32$ \\
Sensing-BS density $\lambda_b$ & $1/100~\mathrm{km}^{-2}$ \\
Minimum sensing range $r_{\min}$ & $100$~m \\
Physical target RCS $\sigma_m$ & $5~\mathrm m^2$ \\
Coherent sensing interval $T_s$ & $1$~s \\
\bottomrule
\end{tabular}
\end{table}

\begin{figure*}[!t]
\centering
\subfloat[Average private rate.\label{fig:rate_validation_private}]{%
\includegraphics[width=0.475\textwidth]{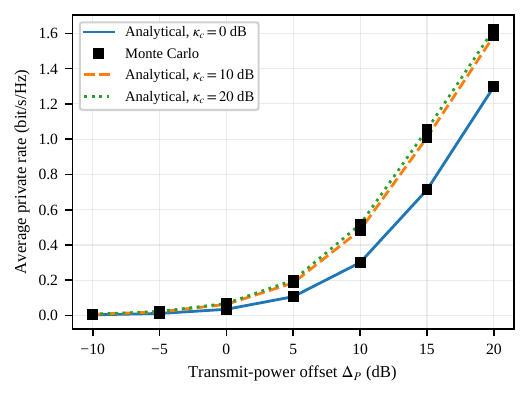}}
\hfill
\subfloat[Supported common rate.\label{fig:rate_validation_common}]{%
\includegraphics[width=0.475\textwidth]{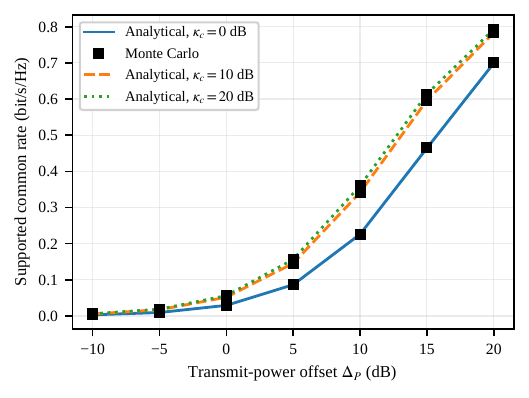}}
\caption{Reference-port analytical validation versus transmit-power offset.}
\label{fig:rate_validation}
\end{figure*}

\subsection{Averaging over the Angular Scheduling Geometry}
\label{sec:angular_avg}

Section~\ref{sec:spatial_avg} conditions on the scheduled angles. We next characterize the candidate-angle process generated by the same finite-window PPP before scheduling.

\begin{lemma}[Induced angular candidate intensity]\label{lem:angular_intensity}
Let the principal-plane coordinate of a candidate user be $x$ and let $\theta=\arctan(x/h_s)$. The restriction $\Phi_u\cap\mathcal D$ projected onto the $x$-axis is an inhomogeneous one-dimensional PPP on $[-r_c,r_c]$ with intensity
\begin{equation}
\lambda_X(x)=2\lambda_u\sqrt{r_c^2-x^2}.
\end{equation}
Under the one-to-one transformation $x=h_s\tan\theta$, its angular intensity is
\begin{align}
\lambda_\Theta(\theta)
&=2\lambda_u h_s\sec^2\theta
\sqrt{r_c^2-h_s^2\tan^2\theta},\nonumber\\
|\theta|&<\theta_{\rm cov},\qquad
\theta_{\rm cov}\triangleq\arctan(r_c/h_s).
\label{eq:angular_intensity}
\end{align}
Moreover, $\int_{-\theta_{\rm cov}}^{\theta_{\rm cov}}\lambda_\Theta(\theta)d\theta=\lambda_u\pi r_c^2$, equal to the mean number of candidate users in the footprint.
\end{lemma}
\begin{proof}
For a vertical strip $[x,x+dx]$, the intersection with $\mathcal D$ has area $2\sqrt{r_c^2-x^2}\,dx+o(dx)$, giving $\lambda_X(x)$. The mapping theorem for PPPs and $dx=h_s\sec^2\theta\,d\theta$ then yield \eqref{eq:angular_intensity}; its integral is preserved by the change of variables.
\end{proof}

\begin{corollary}[Candidate-angle density conditional on the footprint count]\label{cor:angular_density}
Conditioned on $N_u=|\Phi_u\cap\mathcal D|=n>0$, the $n$ candidate locations are i.i.d. uniform on $\mathcal D$. Their angles are therefore i.i.d. with density
\begin{equation}
f_\Theta(\theta)
=\frac{\lambda_\Theta(\theta)}{\lambda_u\pi r_c^2},
\qquad |\theta|<\theta_{\rm cov}.
\label{eq:angular_density}
\end{equation}
\end{corollary}

\begin{figure*}[!t]
\centering
\subfloat[Average private rate.\label{fig:fas_validation_private}]{%
\includegraphics[width=0.475\textwidth]{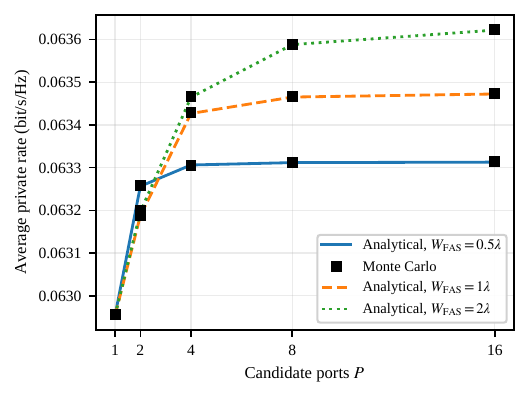}}
\hfill
\subfloat[Supported common rate.\label{fig:fas_validation_common}]{%
\includegraphics[width=0.475\textwidth]{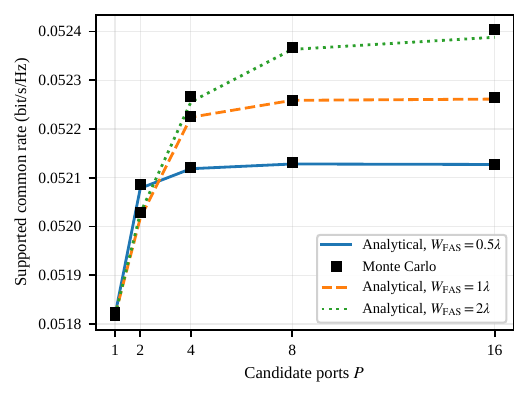}}
\caption{Validation of the FAS-selected analytical bounds versus candidate-port count.}
\label{fig:fas_validation}
\end{figure*}

\paragraph{Full scheduled-angle average}
The farthest-point rule in \eqref{eq:scheduling} is a nonlinear functional of the complete candidate-angle set, so the selected angles are neither independent nor distributed according to \eqref{eq:angular_density}. We therefore generate the finite-window PPP (or, conditional on $N_u=n$, draw $n$ i.i.d. angles from \eqref{eq:angular_density}), apply the scheduler explicitly, and average the resulting deterministic beam gains and performance metrics numerically. No independent-angle approximation is used for the scheduled set.

\begin{figure}[!t]
\centering
\includegraphics[width=\columnwidth]{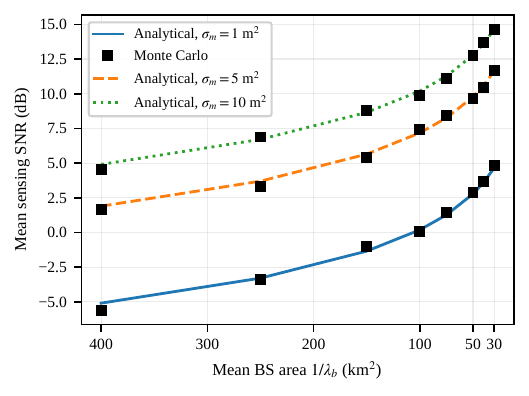}
\caption{Bistatic mean sensing SNR validation versus sensing-BS density.}
\label{fig:sensing_validation}
\end{figure}

\begin{figure}[!t]
\centering
\includegraphics[width=\columnwidth]{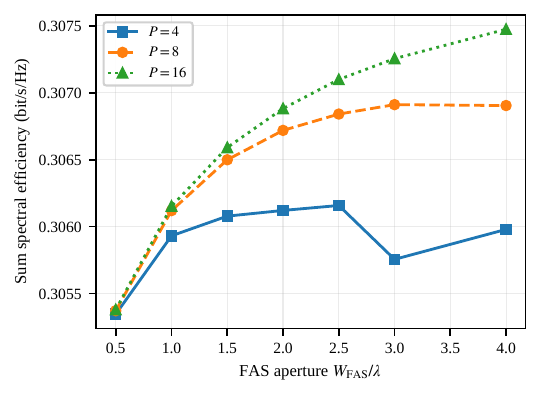}
\caption{Joint impact of FAS aperture and candidate-port count on communication performance.}
\label{fig:fas_design}
\end{figure}

\section{Numerical Validation and Performance Results}
\label{sec:numerical_results}
We first validate the analytical expressions and then compare the architecture with relevant references. The nominal parameters are summarized in Table~\ref{tab:sim_parameters}. The numerical RHS uses a fixed $N_f=25$ feed basis $[\mathbf\Psi]_{n,f}=e^{j\pi(n-1)\xi_f}$ with $\xi_f$ uniformly spaced on $[-0.5,0.5]$; the center feed ($\xi_{f_0}=0$) provides the fixed holographic reference wave. This hardware basis is held unchanged for every geometry and parameter sweep. Monte Carlo communication rates are computed directly from the instantaneous SINRs in \eqref{eq:gamma_common}--\eqref{eq:gamma_rates} and rates in \eqref{eq:instant_rates}, without using the analytical log-moment or mean-interference substitutions. Sensing Monte Carlo trials draw both Rician hops and the nearest-BS distance directly from the model in Section~II-F. Thus, the simulation curves are independent numerical checks of the analytical bounds/mean expressions. All numerical figures use color together with distinct line and marker styles to preserve grayscale readability. In the analytical-validation plots, continuous curves denote analytical expressions and filled-square markers denote independent Monte Carlo results; markers in subsequent performance figures are used only for visual readability unless explicitly stated otherwise.

\subsection{Analytical Validation}
Fig.~\ref{fig:rate_validation} validates the reference-port bounds of Theorems~\ref{thm:private} and~\ref{thm:common} over a $30$-dB transmit-power range for $\kappa_c\in\{0,10,20\}$~dB. At the nominal power offset, the simulated/analytical average private rates are $0.03519/0.03517$, $0.06295/0.06295$, and $0.06835/0.06835$~bit/s/Hz, while the corresponding supported common rates are $0.02930/0.02925$, $0.05181/0.05181$, and $0.05614/0.05614$~bit/s/Hz. At the largest tested offset of $20$~dB, the simulated/analytical private rates are $1.29833/1.29786$, $1.58640/1.58631$, and $1.62335/1.62334$~bit/s/Hz. The near-coincidence verifies the shared-RHS rate bounds over both power and LoS conditions.

Fig.~\ref{fig:fas_validation} validates Lemma~\ref{lem:fas}, Theorem~\ref{thm:fas_private}, Lemma~\ref{lem:fas_crosslog}, and Theorem~\ref{thm:fas_common} for $W_{\rm FAS}\in\{0.5\lambda,\lambda,2\lambda\}$. The numerically evaluated analytical bounds use the full-$\mathbf R_F$ selected-port moments computed by low-discrepancy Gaussian integration and are compared with $3\times10^4$ Monte Carlo channel realizations per operating point. At $P=8$ and $W_{\rm FAS}=2\lambda$, the simulated/analytical private rates are $0.063588/0.063588$~bit/s/Hz, while the supported common rates are $0.052368/0.052363$~bit/s/Hz. The agreement confirms the same-port interference and common-stream coupling under the full Bessel correlation model.

\begin{figure}[!t]
\centering
\includegraphics[width=\columnwidth]{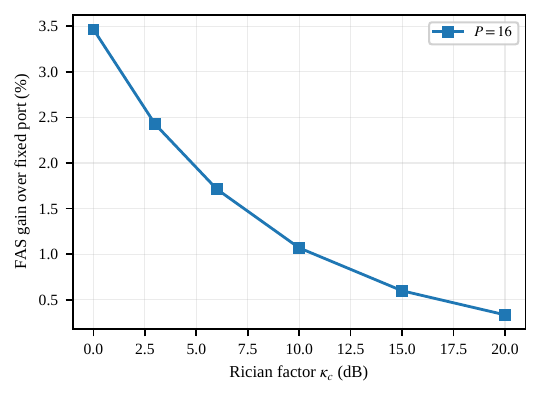}
\caption{Relative FAS gain over fixed-port reception versus the communication-link Rician factor for $P=16$.}
\label{fig:fas_kappa}
\end{figure}

\begin{figure}[!t]
\centering
\includegraphics[width=\columnwidth]{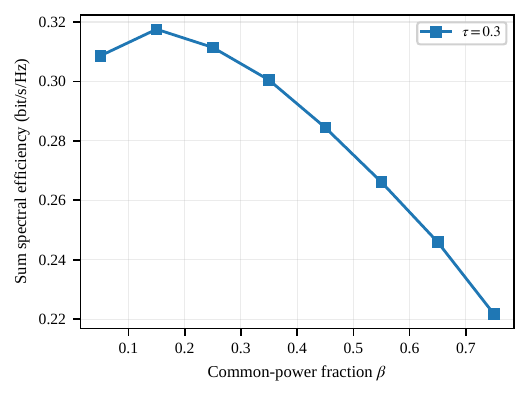}
\caption{Communication sum spectral efficiency versus the RSMA
    common-power fraction $\beta$ at the nominal $\tau=0.3$.}
\label{fig:rsma_beta_comm}
\end{figure}
Fig.~\ref{fig:sensing_validation} validates Theorem~\ref{thm:sensing} for $\sigma_m\in\{1,5,10\}$~m$^2$. At $\lambda_b=1/100~\mathrm{km}^{-2}$, the analytical/Monte Carlo mean sensing SNRs are approximately $0.19/0.07$, $7.18/7.43$, and $10.19/9.89$~dB, respectively. For $\sigma_m=5$~m$^2$, increasing the density from $1/400$ to $1/30~\mathrm{km}^{-2}$ raises the analytical mean SNR from approximately $1.90$ to $11.65$~dB because the nearest target--receiver distance decreases.

\subsection{FAS Spatial Selection and Propagation Regime}
Fig.~\ref{fig:fas_design} shows the joint impact of FAS aperture and candidate-port count. At $W_{\rm FAS}=2\lambda$, the sum spectral efficiencies are approximately $0.30612$, $0.30672$, and $0.30688$~bit/s/Hz for $P=4$, $8$, and $16$, respectively. For $P=16$, enlarging the aperture from $0.5\lambda$ to $4\lambda$ raises the sum spectral efficiency from $0.30538$ to $0.30747$~bit/s/Hz. The small local nonmonotonicity for smaller $P$ is consistent with the oscillatory Bessel correlation structure.

\begin{figure}[!t]
\centering
\includegraphics[width=\columnwidth]{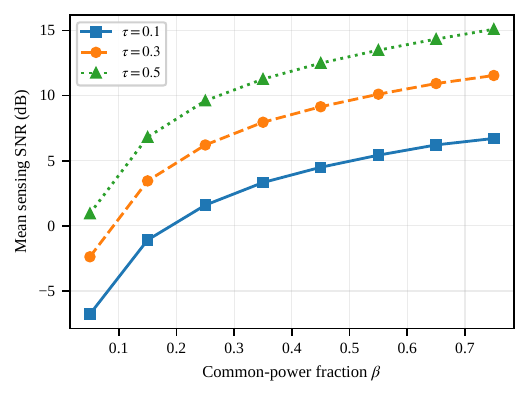}
\caption{Average bistatic sensing SNR versus the RSMA common-power
    fraction $\beta$ for different common-beam steering factors $\tau$.}
\label{fig:beta_sensing}
\end{figure}

\begin{figure}[!t]
\centering
\includegraphics[width=\columnwidth]{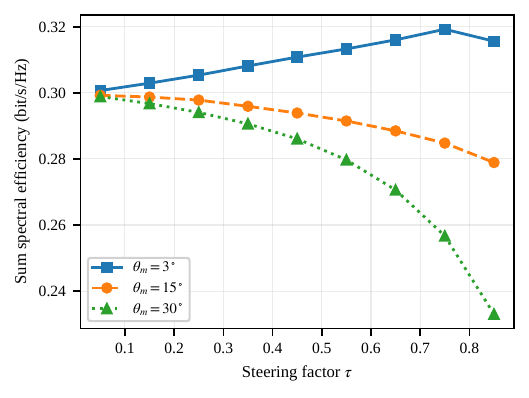}
\caption{Communication sum spectral efficiency versus the common-beam target-steering factor $\tau$ for different target angles.}
\label{fig:tau_geometry_comm}
\end{figure}

Fig.~\ref{fig:fas_kappa} isolates the propagation-regime dependence: the $P=16$ gain over fixed-port reception decreases from about $3.47\%$ at $\kappa_c=0$~dB to $1.07\%$ at $10$~dB and $0.34\%$ at $20$~dB. Thus FAS is most useful when a meaningful locally scattered component remains available for port selection.

\begin{figure}[!t]
\centering
\includegraphics[width=\columnwidth]{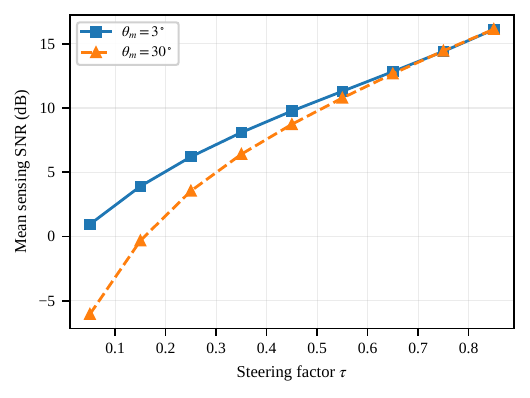}
\caption{Average bistatic sensing SNR versus the common-beam target-steering factor $\tau$ for the two extreme target angles.}
\label{fig:tau_geometry_sensing}
\end{figure}

\begin{figure*}[!t]
\centering
\subfloat[Communication sum spectral efficiency.\label{fig:rhs_benchmark_comm}]{%
\includegraphics[width=0.315\textwidth]{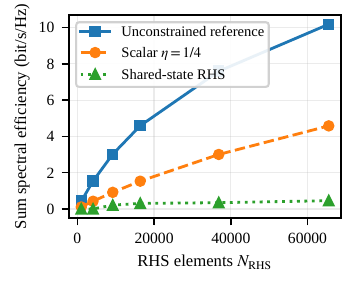}}
\hfill
\subfloat[Average sensing SNR.\label{fig:rhs_benchmark_sensing}]{%
\includegraphics[width=0.315\textwidth]{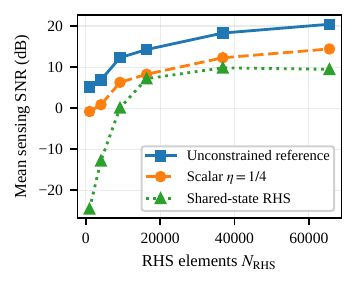}}
\hfill
\subfloat[Realized target-direction gain.\label{fig:rhs_benchmark_gain}]{%
\includegraphics[width=0.315\textwidth]{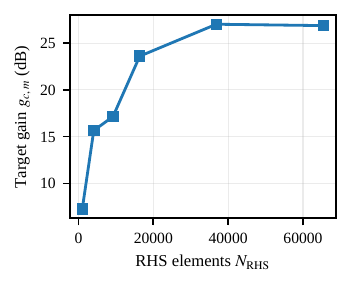}}
\caption{RHS architecture benchmarks versus aperture size. The unconstrained complex-weight curve is a reference beam set, not a direction-wise upper bound.}
\label{fig:rhs_benchmark}
\end{figure*}

\subsection{RSMA Communication--Sensing Coupling}
Fig.~\ref{fig:rsma_beta_comm} shows the nominal $\tau=0.3$ communication curve because the corresponding curves for other $\tau$ values are close. The sum spectral efficiency first increases from $0.3086$ at $\beta=0.05$ to $0.3176$ at $\beta=0.15$, then decreases to $0.2217$ at $\beta=0.75$ as the private-stream power loss dominates. Hence the communication-optimal common-power fraction is interior for the considered geometry.

Fig.~\ref{fig:beta_sensing} shows a much stronger sensing dependence. For $\tau=0.3$, increasing $\beta$ from $0.05$ to $0.75$ raises the mean sensing SNR from approximately $-2.38$ to $11.54$~dB. At $\beta=0.35$, increasing $\tau$ from $0.1$ to $0.3$ and $0.5$ raises the mean sensing SNR from $3.33$ to $7.95$ and $11.26$~dB, respectively. Thus $\beta$ controls the power-domain tradeoff while $\tau$ provides spatial sensing control.
\begin{figure*}[!t]
\centering
\subfloat[Achieved minimum user separation.\label{fig:scheduler_sep}]{%
\includegraphics[width=0.475\textwidth]{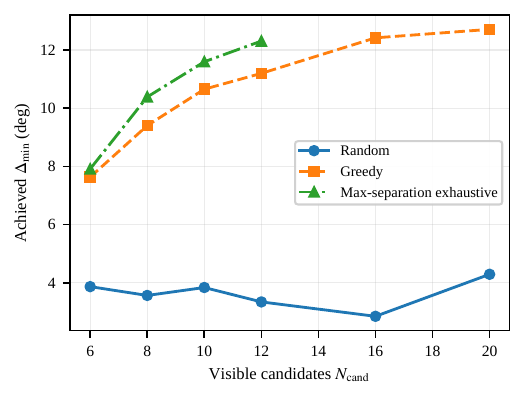}}
\hfill
\subfloat[Minimum pre-mapping ZF gain and Lemma~1 bound.\label{fig:scheduler_gain}]{%
\includegraphics[width=0.475\textwidth]{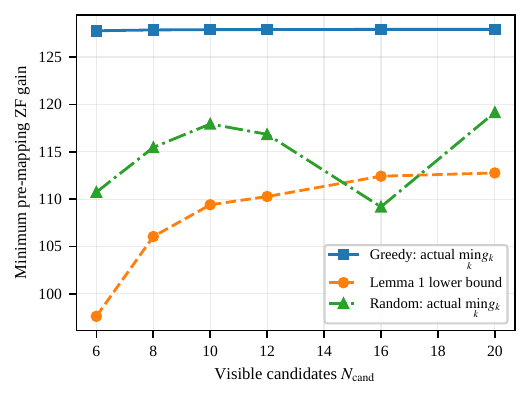}}
\caption{Angular scheduling benchmark.}
\label{fig:scheduler}
\end{figure*}

Fig.~\ref{fig:tau_geometry_comm} examines the communication impact of the common-beam target-steering factor $\tau$ for different target directions. The resulting sum-spectral-efficiency variations remain modest compared with the sensing response, although the trend depends on geometry. For example, at $\theta_m=3^\circ$, the sum spectral efficiency increases from approximately $0.3007$ at $\tau=0.05$ to $0.3156$ at $\tau=0.85$.

Fig.~\ref{fig:tau_geometry_sensing} shows the corresponding sensing behavior for the two extreme target angles. At $\tau=0.05$, the mean sensing SNRs are approximately $0.95$ and $-6.05$~dB for $\theta_m=3^\circ$ and $30^\circ$, respectively. As $\tau$ increases, explicit target steering progressively dominates the common-beam orientation, and both cases approach approximately $16.1$~dB at $\tau=0.85$. This shows that increasing $\tau$ can largely overcome the initial geometry mismatch in the sensing direction.

\subsection{RHS Hardware Benchmarks}
Fig.~\ref{fig:rhs_benchmark} compares the unconstrained complex-weight reference, the scalar $\eta=1/4$ proxy, and the shared-state RHS with the same fixed $25$-feed hardware basis at every aperture size. At $N_{\rm RHS}=16384$, the sum spectral efficiencies are approximately $4.59$, $1.53$, and $0.307$~bit/s/Hz, respectively. The gap quantifies the cost of enforcing one shared amplitude state and a fixed feed network rather than independently synthesized complex aperture weights.

At the same aperture, Fig.~\ref{fig:rhs_benchmark}(b) gives mean sensing SNRs of approximately $14.28$, $8.26$, and $7.18$~dB. At $N_{\rm RHS}=1024$, the shared-state realization falls to $-24.54$~dB because the fixed feed basis and shared recording pattern couple poorly to the small aperture, illustrating why a scalar efficiency factor cannot represent architecture- and direction-dependent realizability effects.

The linear target gain $g_{c,m}$ underlying Fig.~\ref{fig:rhs_benchmark}(c) increases from about $5.34$ at $N_{\rm RHS}=1024$ to $503.0$ at $N_{\rm RHS}=36864$, before decreasing slightly to $486.3$ at $65536$. Hence off-design target gain need not vary monotonically with aperture size even when the hardware feed basis is fixed.

\subsection{Scheduling Geometry}

Fig.~\ref{fig:scheduler} evaluates the angular scheduling rule and Lemma~\ref{lem:sched_gain}. To remain within the lemma's stated $\theta_{\rm sch}\leq30^\circ$ domain, this experiment considers a $300$-km footprint at $h_s=550$~km and averages over random candidate sets. As shown in Fig.~\ref{fig:scheduler}(a), the greedy farthest-point rule consistently produces substantially larger minimum angular separations than random $K$-user selection and remains close to the exhaustive max-separation benchmark over the entire tested range $N_{\rm cand}=6$--$20$. For example, at $N_{\rm cand}=10$, the average achieved separations are approximately $10.49^\circ$,
$3.48^\circ$, and $11.40^\circ$ for greedy, random, and exhaustive selection, respectively. As the candidate population increases, the exhaustive benchmark continues to provide the largest separation, reaching approximately $14.60^\circ$ at $N_{\rm cand}=20$, compared with $12.65^\circ$ for the greedy rule.
\begin{figure*}[!t]
\centering
\subfloat[Communication sum spectral efficiency.\label{fig:bistatic_comm}]{%
\includegraphics[width=0.475\textwidth]{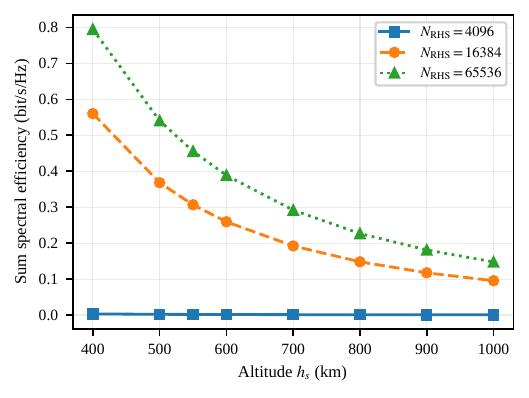}}
\hfill
\subfloat[Bistatic versus favorable monostatic sensing.\label{fig:bistatic_sensing}]{%
\includegraphics[width=0.475\textwidth]{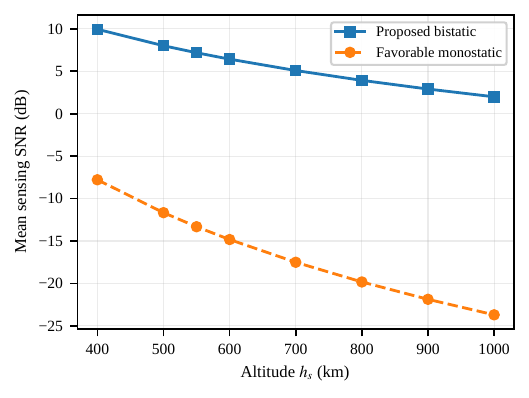}}
\caption{LEO-altitude and sensing-architecture benchmark.}
\label{fig:bistatic_benchmark}
\end{figure*}

Fig.~\ref{fig:scheduler}(b) compares the actual minimum principal-plane pre-mapping ZF gain with the analytical lower bound of Lemma~\ref{lem:sched_gain}. Because the scheduler lemma is derived for the $N_h=N_x$ steering Gram matrix, the plotted $g_k$ is the principal-plane quantity; the corresponding separable full-aperture LoS directional gain is $N_y g_k$ by \eqref{eq:upa_gain}. Although Fig.~\ref{fig:scheduler}(a) reports angular separation in degrees, all angles are converted to radians when evaluating the lemma. The greedy scheduler maintains $\min_k g_k\approx127.69$--$127.92$ over the tested candidate populations, while the lemma provides a conservative but strictly positive guarantee that increases from approximately $96.37$ to $112.63$. Random scheduling yields a substantially lower and more variable minimum ZF gain. These results confirm the role of the proposed scheduler as a geometry-conditioning mechanism that improves user angular separation and preserves favorable pre-mapping ZF gains; it is not claimed to solve the post-mapping sum-rate maximization problem.

\subsection{Bistatic Architecture and LEO Geometry}
Fig.~\ref{fig:bistatic_benchmark}(a) shows the expected communication degradation with altitude. For $N_{\rm RHS}=16384$, increasing $h_s$ from $400$ to $1000$~km reduces the sum spectral efficiency from approximately $0.560$ to $0.095$~bit/s/Hz, while larger apertures partially offset the path-loss penalty.

Fig.~\ref{fig:bistatic_benchmark}(b) compares the proposed nearest-ground-receiver architecture with a deliberately favorable monostatic reference in which the echo returns to the satellite, direct-path/self-interference is neglected, and the same $N_{\rm RHS}$-element aperture receives ideal coherent gain. Its mean SNR is
\begin{equation}
\bar\gamma_m^{\rm mono}
=\frac{\widetilde P_c\zeta_m G_p}{\sigma_b^2}\,\beta_{S,m}^2\,
\frac{\kappa_t g_{c,m}+1}{\kappa_t+1}\,N_{\rm RHS}.
\label{eq:mono_benchmark}
\end{equation}
For $N_{\rm RHS}=16384$ and $h_s=550$~km, the bistatic and monostatic mean SNRs are approximately $7.18$ and $-13.32$~dB, respectively, a $20.5$~dB advantage. The gap increases from about $17.7$~dB at $400$~km to $25.7$~dB at $1000$~km because the monostatic echo experiences the LEO-scale propagation distance twice whereas the bistatic return hop remains terrestrial.

\section{Conclusion}
\label{sec:conclusion}
This paper developed an ergodic performance framework for RSMA-enabled bistatic LEO-ISAC with a shared-state multi-feed RHS and FAS users. Exact realized self, leakage, and target gains were retained in conservative private/common rate bounds, including same-port FAS selection coupling under a full Bessel port-correlation model, while a closed-form mean bistatic sensing SNR was derived for nearest-ground-receiver association and extended to footprint/scheduling geometry. Monte Carlo validation confirmed the communication bounds and mean sensing expression. The results show that FAS gains are largest in scattering-rich regimes, scalar RHS-efficiency models can miss strong architecture- and direction-dependent effects, and realized target gain need not scale monotonically with aperture size. For $N_{\rm RHS}=16384$, nearest-BS bistatic reception retains a $17.7$--$25.7$~dB mean SNR advantage over a favorable monostatic reference across $400$--$1000$~km.

\bibliographystyle{IEEEtran}
\bibliography{Reference}
\end{document}